\documentclass[11pt]{article}

\usepackage{amsmath}
\usepackage{amsthm}
\usepackage{longtable}
\usepackage{multirow}
\usepackage{amssymb}
\usepackage{algorithm}
\usepackage{subfig}
\usepackage{color}
\usepackage[english]{babel}
\usepackage{graphicx}
\usepackage{wrapfig,epsfig}
\usepackage{epstopdf}
\usepackage{url}
\usepackage{graphicx}
\usepackage{color}
\usepackage{epstopdf}
\usepackage{algpseudocode}
\usepackage{scrextend}
\usepackage[T1]{fontenc}
\usepackage{bbm}
\usepackage{comment}

\usepackage{tikz}
\usepackage{hyperref}  
\hypersetup{colorlinks=true,citecolor=blue,linkcolor=blue} 
\usetikzlibrary{arrows}
\usepackage[margin=1in]{geometry}
\graphicspath{{./figs/}}

\newcommand{\pr}[1]{\text{\normalfont Pr}\normalfont\lbrack #1 \rbrack} 
\newcommand{\ex}[1]{\mathbb{E}\normalfont\lbrack #1 \rbrack}
\newcommand{\bpr}[1]{\text{\normalfont Pr}\normalfont \left[#1 \right]} 
\newcommand{\bex}[1]{\mathbb{E}\normalfont \left[#1 \right]}
\newcommand{\1}{\mathbf{1}}

\newtheorem{theorem}{Theorem}[section]
\newtheorem{lemma}[theorem]{Lemma}
\newtheorem{definition}[theorem]{Definition}

\newtheorem{proposition}[theorem]{Proposition}
\newtheorem{corollary}[theorem]{Corollary}

\newtheorem{fact}[theorem]{Fact}

\newtheorem{claim}[theorem]{Claim}

\newcommand{\tail}[2]{#1_{\overline{[#2]}}}

\newcommand{\wh}{\widehat}
\newcommand{\wt}{\widetilde}
\newcommand{\ov}{\overline}
\newcommand{\eps}{\epsilon}

\newcommand{\R}{\mathbb{R}}

\renewcommand{\varepsilon}{\epsilon}
\renewcommand{\tilde}{\wt}
\renewcommand{\hat}{\wh}

\DeclareMathOperator*{\median}{median}

\DeclareMathOperator{\OPT}{OPT}

\DeclareMathOperator{\poly}{poly}

\DeclareMathOperator{\nnz}{nnz}

\DeclareMathOperator{\rank}{rank}
\DeclareMathOperator{\trank}{trank}

\DeclareMathOperator{\vect}{vec}

\DeclareMathOperator{\ati}{ati}
\DeclareMathOperator{\sym}{mon}

\DeclareMathOperator{\reg}{reg}

\makeatletter
\newcommand*{\RN}[1]{\expandafter\@slowromancap\romannumeral #1@}
\makeatother

\usepackage{lineno}

\date{}
\title{Optimal Sketching for Kronecker Product Regression \\and Low Rank Approximation\thanks{A preliminary version of this paper appeared in NeurIPS 2019.}}

\author{
  Huaian Diao\thanks{\texttt{hadiao@nenu.edu.cn}. 
  Northeast Normal University.}
  \quad
  Rajesh Jayaram\thanks{\texttt{rkjayara@cs.cmu.edu}.
  Carnegie Mellon University. Rajesh Jayaram would like to thank support from the Office of Naval Research (ONR) grant N00014-18-1-2562. This work was partly done while Rajesh Jayaram was visiting the Simons Institute for the Theory of Computing.}
  \quad
  Zhao Song\thanks{\texttt{zhaosong@uw.edu}.
  University of Washington. This work was partly done while Zhao Song was visiting the Simons Institute for the Theory of Computing.}
  \quad
  Wen Sun\thanks{\texttt{sun.wen@microsoft.com}.
  Microsoft Research New York.}
  \quad
  David P. Woodruff\thanks{\texttt{dwoodruf@cs.cmu.edu}.
  Carnegie Mellon University. David Woodruff would like to thank support from the Office of Naval Research (ONR) grant N00014-18-1-2562. This work was also partly done while David Woodruff was visiting the Simons Institute for the Theory of Computing.}
}

\begin{document}

\begin{titlepage}
  \maketitle
  \begin{abstract}

We study the Kronecker product regression problem, in which the design matrix is a Kronecker product of two or more matrices. Formally, for $p \in [1,2]$, given $A_i \in \R^{n_i \times d_i}$ for $i=1,2,\dots,q$  where $n_i \gg d_i$ for each $i$, and $b \in \R^{n_1 n_2 \cdots n_q}$, the goal is to find $x \in \R^{d_1 \cdots d_q}$ such that for some $\eps>0$ we have
\[	\| \left( A_1 \otimes A_2 \otimes \cdots \otimes A_q\right)x - b\|_p \leq (1+\epsilon) \min_{x'} \|\left( A_i \otimes A_2 \otimes \cdots \otimes A_q\right)x' - b \|_p \]
 Recently, Diao, Song, Sun, and Woodruff (AISTATS, 2018) gave an algorithm which solves the above problem in time faster than forming the Kronecker product $ A_i \otimes A_2 \otimes \cdots \otimes A_q\in \R^{n_1 \cdots n_q \times d_1 \cdots d_q}$. Specifically, for $p=2$ they achieve a running time of $O(\sum_{i=1}^q   \texttt{nnz}(A_i) + \texttt{nnz}(b))$, where $ \texttt{nnz}(A_i)$ is the number of non-zero entries in $A_i$. Note that $\texttt{nnz}(b)$ can be as large as $\Theta(n_1 \cdots n_q)$. For $p=1,$ $q=2$ and $n_1 = n_2$, they achieve a worse bound of $O(n_1^{3/2} \text{poly}(d_1d_2) + \texttt{nnz}(b))$. 

In this work, we provide significantly faster algorithms. For $p=2$, our running time is $O(\sum_{i=1}^q   \texttt{nnz}(A_i) )$, which has no dependence on $\texttt{nnz}(b)$.  For $p<2$, our running time is $O(\sum_{i=1}^q   \texttt{nnz}(A_i) + \texttt{nnz}(b))$, which matches the prior best running time for $p=2$.  We also consider the related all-pairs regression problem, where given $A \in \R^{n \times d}, b \in \R^n$,  we want to solve $\min_{x \in \R^d} \|\bar{A}x - \bar{b}\|_p$, where $\bar{A} \in \R^{n^2 \times d}, \bar{b} \in \R^{n^2}$ consist of all pairwise differences of the rows of $A,b$. We give an $O(\texttt{nnz}(A))$ time algorithm for $p \in[1,2]$, improving the $\Omega(n^2)$ time required to form $\bar{A}$. Finally, we initiate the study of Kronecker product low rank and low $t$-rank approximation, where the goal is to output a low rank (or low $t$-rank) approximation to a Kronecker product matrix $\mathcal{A} = A_1 \otimes A_2 \otimes \cdots \otimes A_q$. For input $A_1 , A_2 , \dots , A_q$, we give algorithms which run in $O(\sum_{i=1}^q  \texttt{nnz}(A_i))$ time, which is much faster than computing $\mathcal{A}$.

  \end{abstract}
  \thispagestyle{empty}
\end{titlepage}

\tableofcontents
 \newpage

			\section{Introduction}


In the $q$-th order Kronecker product regression problem, one is 
given matrices $A_1,A_2,\dots,A_q$, where $A_i \in \mathbb{R}^{n_i \times d_i}$,
as well as a vector $b \in \mathbb{R}^{n_1 n_2 \cdots n_q}$, and the goal is to obtain a solution to the optimization problem:
$$\min_{x \in \R^{d_1 d_2 \cdots d_q}} \|(A_1 \otimes A_2 \cdots \otimes A_q)x - b\|_p,$$
where $p \in [1, 2]$, and for a vector $x\in \R^n$ the $\ell_p$ norm is defined by $\|x\|_p = (\sum_{i=1}^n |x_i|^p)^{1/p}$. For $p=2$, this is known as \textit{least squares regression}, and for $p=1$ this is known as  \textit{least absolute deviation regression}. 

Kronecker product regression is a special case of ordinary 
regression in which the design matrix is highly structured. Namely, the design matrix is
the Kronecker product of two or more smaller matrices. Such Kronecker product matrices
naturally arise in applications such as spline regression, 
signal processing, and multivariate data fitting. We refer the reader
to \cite{vanbook,vanLoanKron,golubVanLoan2013Book} for further background
and applications of Kronecker product regression. As discussed in
\cite{dssw18}, Kronecker product regression also arises in structured
blind deconvolution problems \cite{oh2005}, 
and the bivariate problem of surface fitting and multidimensional 
density smoothing \cite{eilers2006multidimensional}. 

A recent work of Diao, Song, Sun, and Woodruff \cite{dssw18} utilizes \textit{sketching} techniques to
output an $x \in \R^{d_1 d_2 \cdots d_q}$ with objective function at most $(1+\epsilon)$-times larger than optimal, for both
least squares and least absolute deviation Kronecker product regression. 
Importantly,
their time complexity is faster than 
the time needed to explicitly compute the product $A_1 \otimes \cdots \otimes A_q$. We note that sketching itself is a powerful tool for compressing extremely high dimensional data, and has been used in a number of tensor related problems, e.g., \cite{swz16,lhw17,dssw18,swz19,akkpvwz20}.

For least squares regression, the algorithm of \cite{dssw18} achieves
$O(\sum_{i=1}^q  \nnz(A_i) + \nnz(b) + \poly(d/\epsilon))$ time, where 
$\nnz(C)$ for a matrix $C$
denotes the number of non-zero entries of $C$. Note that the focus is on the over-constrained regression setting, when $n_i \gg d_i$ 
for each $i$, and so the goal is to have a small running time dependence on the $n_i$'s. We remark that over-constrained regression has been the focus of a large body of work over the past decade, which primarily attempts to design fast regression algorithms in the big data (large sample size) regime, see, e.g., \cite{mah11,w14} for surveys.

Observe that explicitly forming the matrix $A_1 \otimes \cdots \otimes A_q$ would take
$\prod_{i=1}^q \nnz(A_i)$ time, which can be as large as $\prod_{i=1}^q n_i d_i$, and so the results of \cite{dssw18} offer a large computational
advantage. Unfortunately, since $b \in \R^{n_1 n_2 \cdots n_q}$, we can have $\nnz(b)  = \prod_{i=1}^q n_i$, and therefore $\nnz(b)$ is likely to be the dominant term in the running time. This leaves open the question of whether it is possible to solve this problem in time \textit{sub-linear} in $\nnz(b)$, with a dominant term of $O(\sum_{i=1}^q  \nnz(A_i) )$. 

For least absolute deviation regression, the bounds of \cite{dssw18}
achieved are still an improvement over computing $A_1 \otimes \cdots \otimes A_q$, though
 worse than the bounds for least squares regression. 
The authors focus on $q = 2$ and the special case $n = n_1 = n_2$. Here,
they obtain a running time of $O(n^{3/2} \poly(d_1 d_2/\epsilon) + \nnz(b))$\footnote{We remark that while the $\nnz(b)$ term is not written in the Theorem of \cite{dssw18}, their approach of leverage score sampling from a well-conditioned basis requires one to sample from a well conditioned basis of $[A_1 \otimes A_2,b]$ for a subspace embedding. As stated, their algorithm only sampled from $[A_1 \otimes A_2]$. To fix this omission, their algorithm would require an additional $\nnz(b)$ time to leverage score sample from the augmented matrix.}. This leaves open the question of whether an \textit{input-sparsity} 
$O(\nnz(A_1) + \nnz(A_2) + \nnz(b) +\poly(d_1 d_2/\epsilon))$ time algorithm exists.

\paragraph{All-Pairs Regression}
In this work, we also study the related all-pairs regression problem. Given $A \in \R^{n \times d}, b \in \R^{n}$, the goal is to approximately solve the $\ell_p$ regression problem $\min_x \|\bar{A}x-\bar{b}\|_p$, where $\bar{A} \in \R^{n^2 \times d}$ is the matrix formed by taking all pairwise differences of the rows of $A$ (and $\bar{b}$ is defined similarly). For $p=1$, this is known as the \textit{rank regression estimator}, which has a long history in statistics. It is closely related to the renowned Wilconxon rank test \cite{wang2009weighted}, and enjoys the desirable property of being robust with substantial efficiency gain with respect to heavy-tailed random errors, while maintaining high efficiency for Gaussian errors \cite{wang2009local,wang2009weighted,Wang2018,wang19}. In many ways, it has properties more desirable in practice than that of the Huber M-estimator \cite{Wang2018,wang19personal}. Recently, the all-pairs loss function was also used by \cite{Wang2018} as an alternative approach to overcoming the challenges of tuning parameter selection for the Lasso algorithm.  However, the rank regression estimator is computationally intensive to compute, even for moderately sized data, since the standard procedure (for $p = 1$) is to solve a linear program with $O(n^2)$ constraints. 
 In this work, we demonstrate the first highly efficient algorithm for this estimator.

\paragraph{Low-Rank Approximation}
 Finally, in addition to regression, we extend our techniques to the Low Rank Approximation (LRA) problem. Here, given a large data matrix $A$, the goal is to find a low rank matrix $B$ which well-approximates $A$. LRA is useful in numerous applications, such as compressing massive datasets to their primary components for storage, denoising, and fast matrix-vector products. 
Thus, designing fast algorithms for approximate LRA has become a large and highly active area of research; see \cite{w14} for a survey. For an incomplete list of recent work using sketching techniques for LRA, see \cite{cw13,mm13,nn13,bw14,cw15soda,cw15focs,rsw16,bwz16,swz17,mw17,cgklpw17,lhw17,swz18,bw18,swz19c,swz19,swz19b,bbbklw19,ivww19} and the references therein. 
 
Motivated by the importance of LRA, we initiate the study of low-rank approximation of Kronecker product matrices. Given $q$ matrices $A_1, \cdots, A_q$ where $A_i \in \R^{n_i \times d_i}$, $n_i \gg d_i$,  $A=\otimes_{i=1}^q A_i$, the goal is to output a rank-$k$ matrix $B \in \R^{n\times  d}$ such that $\| B - A \|_F^2 \leq (1 + \eps)\OPT_k$, where $\OPT_k$ is the cost of the best rank-$k$ approximation, $n = n_1 \cdots n_q$, and $d = d_1 \cdots d_q$. Here $\|A\|_F^2 = \sum_{i,j}A_{i,j}^2$. The fastest general purpose algorithms for this problem run in time $O(\nnz(A) + \poly(d k/\eps))$ \cite{cw13}.  However, as in regression, if $A = \otimes_{i=1}^q A_i$, we have $\nnz(A) = \prod_{i=1}^q \nnz(A_i)$, which grows very quickly. Instead, one might also hope to obtain a running time of $O( \sum_{i=1}^q \nnz(A_i) + \poly(d k/\eps))$.

\subsection{Our Contributions}
Our main contribution is an input sparsity time $(1+\epsilon)$-approximation algorithm to Kronecker product regression for every $p \in [1,2]$, and $q\geq 2$. Given $A_i \in \R^{n_i \times d_i}$, $i = 1, \ldots, q$, and $b \in \R^n$ where $n = \prod_{i=1}^q n_i$, together with accuracy parameter $\epsilon \in (0,1/2)$ and failure probability $\delta > 0$, the goal is to output a vector $x' \in \R^d$ where $d = \prod_{i=1}^q d_i$ such that 
\[\| (A_1 \otimes \cdots \otimes A_q) x' - b \|_p \leq (1+\epsilon) \min_x \| (A_1 \otimes \cdots \otimes A_q) x - b \|_p \]
 holds with probability at least $1 - \delta$. For $p=2$, our algorithm runs in $\wt{O}\left(\sum_{i=1}^q \nnz(A_i) ) + \poly(d \delta^{-1}/\eps )\right)$ time.\footnote{For a function $f(n,d,\eps,\delta)$, $\wt{O}(f) = O(f \cdot \poly(\log n))$} Notice that this is \textit{sub-linear} in the input size, since it does not depend on $\nnz(b)$. For $p<2$, the running time is $\wt{O}\left((\sum_{i=1}^q \nnz(A_i)  +\nnz(b) + \poly(d /\eps ))\log(1/\delta)\right)$. Specifically, we prove the following two Theorems: 

\begin{theorem}[Restatement of Theorem \ref{thm:l2_Kronregression}, Kronecker product $\ell_2$ regression]
	Let $D \in \R^{n \times n}$ be the diagonal row sampling matrix generated via Proposition \ref{prop:leveragesamplemain}, with $m = \Theta(1/(\delta\eps^2))$ non-zero entries, and let $A = \otimes_{i=1}^q A_i$, where $A_i \in \R^{n_i \times d_i}$, and $b \in \R^n$, where $n = \prod_{i=1}^q n_i$ and $d=\prod_{i=1}^q d_i$. Then we have let $\hat{x} = \arg \min_{x \in \R^d} \|D A x - Db\|_2$, and let $x^* = \arg \min_{x' \in \R^d} \|Ax-b\|_2$. Then with probability $1-\delta$, we have
	\[ \|A \hat{x} - b\|_2 \leq (1+\eps)  \|Ax^*-b\|_2\]
	Moreover, the total runtime requires to compute $\hat{x}$ is 
	\begin{align*}
	\tilde{O} \left( \sum_{i=1}^q \nnz(A_i) + \poly( dq / (\delta\eps) ) \right).
	\end{align*}
\end{theorem}

\begin{theorem}[Restatement of Theorem \ref{thm:l1_regression}, Kronecker product $\ell_p$ regression]
	Fix $1 \leq p < 2$. Then for any constant $q = O(1)$, given matrices $A_1, A_2, \cdots, A_q$, where $A_i \in \R^{n_i \times d_i}$, let $n = \prod_{i=1}^q n_i$, $d = \prod_{i=1}^q d_i$.
	Let $\hat{x} \in \R^d$ be the output of Algorithm \ref{alg:l1}. Then 
	\[\|  (A_1 \otimes A_2 \otimes \cdots \otimes A_q) \hat{x} - b \|_p  \leq (1+ \eps) \min_{x \in \R^n}\|  (A_1 \otimes A_2 \otimes \cdots \otimes A_q) x -b \|_p	\]
	holds with probability at least $1-\delta$. In addition, our algorithm takes 
	\begin{align*}
	\wt{O} \left( \left( \sum_{i=1}^q \nnz(A_i)  + \nnz(b) + \poly(d \log(1/\delta)/\eps) \right) \log(1/\delta) \right)
	\end{align*}
	time to output $\wh{x} \in \R^d$. 
\end{theorem}


Observe that in both cases, this running time is significantly faster than the time to write down $A_1 \otimes \cdots \otimes A_q$. For $p=2$, up to logarithmic factors, the running time is the same as the time required to simply read each of the $A_i$. Moreover, in the setting  $p<2$, $q=2$ and $n_1 = n_2$ considered in \cite{dssw18}, our algorithm offers a substantial improvement over their running time of $O(n^{3/2} \poly(d_1 d_2/\eps))$.


We empirically evaluate our Kronecker product regression algorithm on exactly the same datasets as those used in \cite{dssw18}. For $p \in \{1,2\}$, the accuracy of our algorithm is nearly the same as that of \cite{dssw18}, while the running time is significantly faster. 


For the all-pairs (or rank) regression problem, we first note that for $A \in \R^{n \times d}$, one can rewrite $\bar{A} \in \R^{n^2 \times d}$ as the difference of Kronecker products $\bar{A} = A \otimes \1^n - \1^n \otimes A$ where $\1^n \in \R^n$ is the all ones vector. Since $\bar{A}$ is not a Kronecker product itself, our earlier techniques for Kronecker product regression are not directly applicable. Therefore, we utilize new ideas, in addition to careful sketching techniques, to obtain an $\wt{O}(\nnz(A) + \poly(d/\eps))$ time algorithm for $p \in [1,2]$, which improves substantially on the $O(n^2d)$ time required to even compute $\bar{A}$, by a factor of at least $n$. Formally, our result for all-pairs regression is as follows:

\begin{theorem}[Restatement of Theorem \ref{thm:allpairsmain}]
	Given $A \in \R^{n \times d}$ and $b \in \R^n$, for $p \in [1,2]$ there is an algorithm for the All-Pairs Regression problem that outputs $\hat{x} \in \R^d$ such that with probability $1-\delta$ we have
	\[ 	\|\bar{A}\hat{x} - \bar{b}\|_p \leq (1+\eps) \min_{x \in \R^d}\|\bar{A}x -\bar{b}\|_p	\]
	Where  $\bar{A} = A \otimes \1 -\1 \otimes A \in \R^{n^2 \times d}$ and $\bar{b} = b \otimes \1 - \1 \otimes b \in \R^{n^2}$.
	For $p <2$, the running time is $\wt{O}(nd + \poly(d/(\eps\delta)))$, and for $p=2$ the running time is $O(\nnz(A) + \poly(d/(\eps\delta)) )$.
\end{theorem}

 Our main technical contribution for both our $\ell_p$ regression algorithm and the rank regression problem is a novel and highly efficient $\ell_p$ sampling algorithm. Specifically, for the rank-regression problem we demonstrate, for a given $x \in \R^d$, how to independently sample $s$ entries of a vector $\bar{A}x = y \in \R^{n^2}$ from the $\ell_p$ distribution $(|y_1|^p/\|y\|_p^p,\dots,|y_{n^2}|^p/\|y\|_p^p)$ in $\wt{O}(nd + \poly(ds))$ time. For the $\ell_p$ regression problem, we demonstrate the same result when $y = (A_1 \otimes \cdots \otimes A_q) x - b \in \R^{n_1 \cdots n_q}$, and in time $\wt{O}(\sum_{i=1}^q \nnz(A_i) + \nnz(b) + \poly(ds))$.  
 This result allows us to sample a small number of rows of the input to use in our sketch. Our algorithm draws from a large number of disparate sketching techniques, such as the dyadic trick for quickly finding heavy hitters \cite{cormode2005improved,knpw11,larsen2016heavy,ns19}, and the precision sampling framework from the streaming literature \cite{ako10}.


For the Kronecker Product Low-Rank Approximation (LRA) problem, we give an input sparsity $O(\sum_{i=1}^q \nnz(A_i) + \poly(dk/\eps))$-time algorithm which computes a rank-$k$ matrix $B$ such that $\|B - \otimes_{i=1}^q A_i\|_F^2 \leq (1 + \eps) \min_{\rank-k~B'} \|B' - \otimes_{i=1}^q A_i\|_F^2$. Note again that the dominant term $\sum_{i=1}^q \nnz(A_i)$ is substantially smaller than the $\nnz(A) = \prod_{i=1}^q \nnz(A_i)$ time required to write down the Kronecker Product $A$, which is also the running time of state-of-the-art general purpose LRA algorithms \cite{cw13,mm13,nn13}. Thus, our results demonstrate that substantially faster algorithms for approximate LRA are possible for inputs with a Kronecker product structure.

\begin{theorem}[Restatement of Theorem \ref{thm:LRAmain}]
	For any constant $q \geq 2$, there is an algorithm which runs in time $O(\sum_{i=1}^q \nnz(A_i) + d\poly(k/\eps))$ and outputs a rank $k$-matrix $B$ in factored form such that $\| B - A \|_{F} \leq (1+\eps) \OPT_k$ with probability $9/10$. 
\end{theorem}

 Our technical contributions employed towards the proof of Theorem \ref{thm:LRAmain} involve demonstrating that useful properties of known sketching matrices hold also for the Kronecker product of these matrices. Specifically, we demonstrate the Kronecker products of the well-known \textit{count-sketch} matrices satisfy the property of being \textit{Projection Cost Preserving Sketches} (PCP). By properties of the Kronecker product, we can quickly apply such a sketching matrix to the input matrix $A$, and the PCP property will allow us to bound the cost of the best low rank approximation obtained via the sketch. 


In addition, motivated by \cite{l00}, we use our techniques to solve the low-$\trank$ approximation problem, where we are given an arbitrary matrix
$A \in \R^{n^q \times n^q}$, and the goal is to output a $\trank$-$k$ matrix $B \in \R^{n^q \times n^q}$ such that $\| B - A \|_F$ is minimized. Here, the $\trank$ of a matrix $B$ is the smallest integer $k$ such that $B$ can be written as a summation of $k$ matrices, where each matrix is the Kronecker product of $q$ matrices with dimensions $n \times n$. Compressing a matrix $A$ to a low-$\trank$ approximation yields many of the same benefits as LRA, such as compact representation, fast matrix-vector product, and fast matrix multiplication, and thus is applicable in many of the settings where LRA is used. Using similar sketching ideas, we provide an $O(\sum_{i=1}^q \nnz(A_i) + \poly(d_1\cdots d_q/\eps))$ time algorithm for this problem under various loss functions. Our results for low-$\trank$ approximation can be found in Section \ref{app:lowTrank}.


%

\section{Preliminaries}
\paragraph{Notation}
For a tensor $A \in \R^{n_1 \times n_2 \times n_3}$, we use $\| A \|_p$ to denote the entry-wise $\ell_p$ norm of $A$, i.e., $\|A \|_p = (\sum_{i_1} \sum_{i_2} \sum_{i_3} |A_{i_1,i_2,i_3}|^p )^{1/p}$. For $n \in \mathbb{N}$, let $[n] = \{1,2,\dots,n\}$. For a matrix $A$, let $A_{i,*}$ denote the $i$-th row of $A$, and $A_{*,j}$ the $j$-th column. For $a,b \in \R$ and $\eps \in (0,1)$, we write $a = (1 \pm \eps)b$ to denote $(1 - \eps)b \leq a \leq (1+\eps)b$. 
We now define various sketching matrices used by our algorithms. 



\paragraph{Stable Transformations}
We will utilize the well-known $p$-stable distribution, $\mathcal{D}_p$ (see \cite{nolan2009stable, i06} for further discussion), which exist for $p \in (0,2]$. For $p \in (0,2)$, $X \sim \mathcal{D}_p$ is defined by its characteristic function ${\mathbb{E}}_{X}[ \exp( \sqrt{-1} t X ) ] = \exp(-|t|^p)$, and can be efficiently generated to a fixed precision \cite{nolan2009stable,kane2010exact}. For $p=2$, $\mathcal{D}_2$ is just the standard Gaussian distribution, and for $p=1$, $\mathcal{D}_1$ is the \textit{Cauchy} distribution. The distribution $\mathcal{D}_p$ has the property that if $z_1,\dots,z_n \sim D_p$ are i.i.d., and $a \in \R^n$, then $\sum_{i=1}^n z_i a_i \sim z \|a\|_p$ where $\|a\|_p = (\sum_{i=1}^n |a_i|^p)^{1/p}$, and $z \sim \mathcal{D}_p$. This property will allow us to utilize sketches with entries independently drawn from $\mathcal{D}_p$ to preserve the $\ell_p$ norm.
\begin{definition}[Dense $p$-stable Transform, \cite{cdmmmw13,sw11}]\label{def:dense_sketch}
Let $p\in [1,2]$. Let $S = \sigma \cdot C \in \mathbb{R}^{m\times n}$, where $\sigma$ is a scalar, and each entry of $C\in\mathbb{R}^{m\times n}$ is chosen independently from $\mathcal{D}_p$. 
\end{definition}

We will also need a sparse version of the above. 
\begin{definition}[Sparse $p$-Stable Transform, \cite{mm13,cdmmmw13}]\label{def:sparse_sketch}
Let $p\in [1,2]$. Let $\Pi = \sigma \cdot S C\in \mathbb{R}^{m\times n}$, where $\sigma$ is a scalar, $S\in \mathbb{R}^{m\times n}$ has each column chosen independently and uniformly from the $m$ standard basis vectors of $\mathbb{R}^{m}$, and $C\in \mathbb{R}^{n\times n}$ is a diagonal matrix with diagonals chosen independently from the standard $p$-stable distribution. For any matrix $A\in \mathbb{R}^{n\times d}$, $\Pi A$ can be computed in $O(\nnz(A))$ time. 
\end{definition}

One nice property of $p$-stable transformations is that they provide \textit{low-distortion $\ell_p$ embeddings}.
\begin{lemma}[Theorem 1.4 of \cite{ww19}; see also Theorem 2 and 4 of \cite{mm13} for earlier work \footnote{In discussion with the authors of these works, the original $O((d \log d)^{1/p})$ distortion factors stated in these papers should be replaced with $O(d \log d)$; as we do not optimize the poly$(d)$ factors in our analysis, this does not affect our bounds.} ]\label{lem:lowdistembed}
	Fix $A \in \R^{n \times d}$, and let $S \in \R^{k \times n}$ be a sparse or dense $p$-stable transform for $p \in [1,2)$, with $k =\Theta(d^2/\delta)$. Then with probability $1-\delta$, for all $x \in \R^d$:
	\[\|Ax\|_p \leq \|SAx\|_p \leq O(d \log d) \|Ax\|_p \]
\end{lemma}
We simply call a matrix $S \in \R^{k \times n}$ a low distortion $\ell_p$ embedding for $A \in \R^{n \times d}$ if it satisfies the above inequality for all $x \in \R^d$.

\paragraph{Leverage Scores \& Well Condition Bases.}

We now introduce the notions of $\ell_2$ leverage scores and well-conditioned bases for a matrix $A \in \R^{n \times d}$.
\begin{definition}[$\ell_2$-Leverage Scores, \cite{w14,bss12}]\label{def:l2leverage}
Given a matrix $A \in \R^{n \times d}$, let $A = Q \cdot R$ denote the QR factorization of matrix $A$. For each $i \in [n]$, we define $\sigma_i = \frac{ \| (A R^{-1})_i \|_2^2 }{ \| A R^{-1} \|_F^2 }$, where $(AR^{-1})_i \in \R^d$ is the $i$-th row of matrix $(A R^{-1}) \in \R^{n \times d}$. We say that $\sigma \in \R^n$ is the $\ell_2$ leverage score vector of $A$.
\end{definition}





\begin{definition}[$(\ell_p,\alpha,\beta)$ Well-Conditioned Basis, \cite{c05}]\label{def:wellconditioned}
Given a matrix $A \in \R^{n \times d}$, we say $U \in \R^{n \times d}$ is an $(\ell_p,\alpha,\beta)$ well-conditioned basis for the column span of $A$ if the columns of $U$ span the columns of $A$, and if for any $x \in \R^d$, we have $\alpha \| x \|_p \leq \| Ux \|_p \leq \beta \| x \|_p $, 
where $\alpha \leq 1 \leq \beta$. If $\beta/\alpha = d^{O(1)}$, then we simply say that $U$ is an $\ell_p$ well conditioned basis for $A$.
\end{definition}

\begin{fact}[\cite{ww19,mm13}]\label{fact:wellcond}
	Let $A \in \R^{n \times d}$, and let $SA \in \R^{k \times d}$ be a low distortion $\ell_p$ embedding for $A$ (see Lemma \ref{lem:lowdistembed}), where $k = O(d^2/\delta)$. Let $SA = QR$ be the $QR$ decomposition of $SA$. Then $AR^{-1}$ is an $\ell_p$ well-conditioned basis with probability $1-\delta$. 
\end{fact}




{
	\begin{algorithm*}[!h]\caption{Our $\ell_2$ Kronecker Product Regression Algorithm}\label{alg:l2reg}
		\begin{algorithmic}[1]
			\Procedure{$\ell_2$ Kronecker Regression}{$(\{ A_i, n_i, d_i \}_{i \in [q] }, b)$} \Comment{Theorem~\ref{thm:l2_Kronregression}}
			\State $d \leftarrow \prod_{i=1}^q d_i$, $n \leftarrow \prod_{i=1}^q n_i$, $m \leftarrow \Theta(d/(\delta\eps^2))$.
			\State Compute approximate leverage scores $\tilde{\sigma}_i(A_j)$ for all $j \in [q],$ $i \in [n_j]$. \Comment{Proposition~\ref{prop:leveragescoresCW}}
			\State Construct diagonal leverage score sampling matrix $D \in \R^{n \times n}$, with $m$ non-zero entries \Comment{Proposition~\ref{prop:leveragesamplemain}}
			\State Compute (via the psuedo-inverse) \\ \hspace{30mm} $ \wh{x} = \arg \min_{x \in \R^d} \| D(A_1 \otimes A_2 \otimes \cdots \otimes A_q)x - Db\|_2 $
			\State \Return $\wh{x}$
			\EndProcedure
		\end{algorithmic}
	\end{algorithm*}
}

\newpage

\section{Kronecker Product Regression}\label{sec:l2reg}
We first introduce our algorithm for $p=2$. Our algorithm for $1 \leq p <2$ is given in Section \ref{sec:lpreg}.  Our regression algorithm for $p=2$ is formally stated in Algorithm \ref{alg:l2reg}. Recall that our input design matrix is $A = \otimes_{i=1}^q A_i$, where $A_i \in \R^{n_i \times d_i}$, and we are also given $b \in \R^{n_1 \cdots n_q}$. Let $n = \prod_{i=1}^q n_i$ and $d = \prod_{i=1}^q d_i$.  The crucial insight of the algorithm is that one can approximately compute the leverage scores of $A$ given only good approximations to the leverage scores of each $A_i$. Applying this fact gives a efficient algorithm for sampling rows of $A$ with probability proportional to the leverage scores. Following standard arguments, we will show that by restricting the regression problem to the sampled rows, we can obtain our desired $(1 \pm \eps)$-approximate solution efficiently.

Our main theorem for this section is stated below. A full proof of the theorem can be found section \ref{subsec:l2reg}.

\begin{theorem}[Kronecker product $\ell_2$ regression]\label{thm:l2_Kronregression}
	Let $D \in \R^{n \times n}$ be the diagonal row sampling matrix generated via Proposition \ref{prop:leveragesamplemain}, with $m = \Theta(d/(\delta\eps^2))$ non-zero entries, and let $A = \otimes_{i=1}^q A_i$, where $A_i \in \R^{n_i \times d_i}$, and $b \in \R^n$, where $n = \prod_{i=1}^q n_i$ and $d=\prod_{i=1}^q d_i$. Then let $\hat{x} = \arg \min_{x \in \R^d} \|D A x - Db\|_2$, and let $x^* = \arg \min_{x' \in \R^d} \|Ax-b\|_2$. Then with probability $1-\delta$, we have
	\[ \|A \hat{x} - b\|_2 \leq (1+\eps)  \|Ax^*-b\|_2.\]
	Moreover, the total running time required to compute $\hat{x}$ is  $\tilde{O}( \sum_{i=1}^q \nnz(A_i) + (dq/(\delta\eps))^{O(1)})$.\footnote{We remark that the exponent of $d$ in the runtime can be bounded by $3$. To see this, first note that the main computation taking place is the leverage score computation from Proposition \ref{prop:leveragescoresCW}. For a $q$ input matrices, we need to generate the leverage scores to precision $\Theta(1/q)$, and thus the complexity from running Proposition \ref{prop:leveragescoresCW} to approximate leverage scores is $O(d^3/q^4)$ by the results of \cite{cw13}. The remaining computation is to compute the pseudo-inverse of a $d/\epsilon^2 \times d$ matrix, which requires $O(d^3/\epsilon^2)$ time, so the additive term in the Theorem  can be replaced with $O(d^3/\epsilon^2 + d^3/q^4)$.}.
\end{theorem}
\subsection{Kronecker Product $\ell_2$ Regression}\label{subsec:l2reg}
We now prove the correctness of our $\ell_2$ Kronecker product regression algorithm. Specifically, we prove Theorem \ref{thm:l2_Kronregression}. To prove correctness, we need to establish several facts about the leverage scores of a Kronecker product. 
\begin{proposition}
	Let $U_i \in \R^{n_i \times d_i}$ be an orthonormal basis for $A_i \in \R^{n_i \times d_i}$. Then $U=\otimes_{i=1}^q U_i$ is an orthonormal basis for $A = \otimes_{i=1}^q A_i$.
\end{proposition}
\begin{proof}
	Note that the column norm of each column of $U$ is the product of column norms of the $U_i$'s, which are all $1$. Thus $U$ has unit norm columns. It suffices then to show that all the singular values of $U$ are $1$ or $-1$, but this follows from the fact that the singular values of $U$ are the product of singular values of the $U_i$'s, which completes the proof. 
\end{proof}
\begin{corollary}\label{cor:leverage}
	Let $A = \otimes_{i=1}^q A_i$, where $A_i \in \R^{n_i \times d_i}$. Fix any $\vec{i} = (i_1,\dots,i_q) \in [n_1] \times [n_2] \times \dots \times [n_q]$, and let $\vec{i}$ index into a row of $A$ in the natural way. Then the $\vec{i}$-th leverage score of $A$ is equal to $\prod_{j=1}^q \sigma_{i_j}(A_j)$, where $\sigma_t(B)$ is the $t$-th leverage score of a matrix $B$.
\end{corollary}
\begin{proof}
	Note $U=\otimes_{i=1}^q U_i$ is an orthonormal basis for $A = \otimes_{i=1}^q A_i$ by the prior Proposition. Now if $U_{\vec{i},*}$ is the $\vec{i}$-th row of $U$, then by fundamental properties of Kronecker products \cite{l00}, we have $\|U_{\vec{i},*}\|_2=\prod_{j=1}^q \|(U_j)_{i_j,*}\|_2$, which completes the proof.	Note here that we used the fact that leverage scores are independent of the choice of orthonormal basis \cite{w14}.
\end{proof}

\begin{proposition}[Theorem 29 of \cite{cw13}]\label{prop:leveragescoresCW}
	Given a matrix $A \in \R^{n \times d}$, let $\sigma \in \R^n$ be the $\ell_2$ leverage scores of $A$ (see definition \ref{def:l2leverage}). Then there is an algorithm which computes values $\tilde{\sigma}_1,\tilde{\sigma}_2,\dots,\tilde{\sigma}_n$ such that $\tilde{\sigma_i} = (1 \pm \eps)\sigma_i$ simultaneously for all $i \in [n]$ with probability $1-1/n^c$ for any constant $c \geq 1$. The runtime is $\tilde{O}(\nnz(A) + d^3/\epsilon^2)$. 
\end{proposition}

\begin{proposition}\label{prop:leveragesamplemain}
	Given $A = \otimes_{i=1}^q A_i$, where $A_i \in \R^{n_i \times d_i}$, there is an algorithm which, with probability $1-1/n^c$ for any constant $c \geq 1$, outputs a diagonal matrix $D \in \R^{n \times n}$ with $m$ non-zeros entries, such that $D_{i,i} =1/(m \tilde{\sigma}_i)$ is non-zero with probability $\tilde{\sigma}_i \in (1 \pm 1/10) \sigma_i(A)$. The time required is $\tilde{O}( \sum_{i=1}^q \nnz(A_i) + \poly(dq/\eps) + mq)$.
\end{proposition}
\begin{proof}
	By Proposition \ref{prop:leveragescoresCW}, we can compute approximate leverage scores of each $A_i$ up to error $\Theta(1/q)$ in time $\tilde{O}(\nnz(A_i) + \poly(d/\eps))$ with high probability. To sample a leverage score from $A$, it suffices to sample one leverage score from each of the $A_i$'s by Corollary \ref{cor:leverage}. The probability that a given row $\vec{i} = (i_1,\dots,i_q) \in [n_1] \times [n_2] \times \dots \times [n_q]$ of $A$ is chosen is $\prod_{j=1}^q \tilde{\sigma}(A_j)_{i_j} = (1 \pm \Theta(1/q))^q \sigma_{\vec{i}}(A) = (1 \pm 1/10)\sigma_{\vec{i}}(A)$ as needed. Obtaining a sample takes $\tilde{O}(1)$ time per $A_i$ (since a random number needs to be generated to $O(\log(n))$-bits of precision in expectation and with high probability to obtain this sample), thus $O(q)$ time overall, so repeating the sampling $M$ times gives the desired additive $mq$ runtime.
\end{proof}

The $q = 1$ version of the following result can be found in \cite{cw13,swz19}.
\begin{proposition}\label{prop:approxmatrixprod}
	Let $D \in \R^{n \times n}$ be the diagonal row sampling matrix generated via Proposition \ref{prop:leveragesamplemain}, with $m = \Theta(1/(\delta\eps^2))$ non-zero entries. Let $A = \otimes_{i=1}^q A_i$ as above, and let $U \in \R^{n \times r}$ be an orthonormal basis for the column span of $A$, where $r = \rank(A)$. Then for any matrix $B$ with $n$ rows, we have 
	\[ \bpr{ \| U^\top D^\top  D B - U^\top  B \|_F \leq \eps \| U \|_F\|B\|_F} \geq 1 - \delta \]
\end{proposition}
\begin{proof}
	By definition of leverage scores and Proposition \ref{prop:leveragesamplemain}, $D$ is a matrix which sample each row $U_{i,*}$ of $U$ with probability at least $(9/10)\|U_{i,*}\|_2/\|U\|_F$. Taking the average of $m$ such rows, we obtain the approximate matrix product result with error $O(1/\sqrt{\delta m})$ with probability $1-\delta$ by Theorem 2.1 of \cite{kannan2017randomized}.
\end{proof}

We now ready to prove the main theorem of this section, Theorem \ref{thm:l2_Kronregression} \\
{\bf Theorem \ref{thm:l2_Kronregression}}~(Kronecker product $\ell_2$ regression){\bf.} {\it
	Let $D \in \R^{n \times n}$ be the diagonal row sampling matrix generated via Proposition \ref{prop:leveragesamplemain}, with $m = \Theta(1/(\delta\eps^2))$ non-zero entries, and let $A = \otimes_{i=1}^q A_i$, where $A_i \in \R^{n_i \times d_i}$, and $b \in \R^n$, where $n = \prod_{i=1}^q n_i$ and $d=\prod_{i=1}^q d_i$. Then we have let $\hat{x} = \arg \min_{x \in \R^d} \|D A x - Db\|_2$, and let $x^* = \arg \min_{x' \in \R^d} \|Ax-b\|_2$. Then with probability $1-\delta$, we have
	\[ \|A \hat{x} - b\|_2 \leq (1+\eps)  \|Ax^*-b\|_2\]
	Moreover, the total runtime requires to compute $\hat{x}$ is 
	\begin{align*}
	\tilde{O} \left( \sum_{i=1}^q \nnz(A_i) + \poly( dq / (\delta\eps) ) \right).
	\end{align*}}
\begin{proof}
	Let $U$ be an orthonormal basis for the column span of $A$. By Lemma 3.3 of \cite{cw09}, we have $\|A(\hat{x} - x^*)\|_2 \leq 2 \sqrt{\eps}\|Ax^* - b\|_2$. Note that while  Lemma 3.3 of \cite{cw09} uses a different sketching matrix $D$ than us, the only property required for the proof of Lemma 3.3 is that $| U^\top D^\top  D B - U^\top  B \|_F \leq \sqrt{\eps/d} \|A\|_F\|B\|_F$ with probability at least $1-\delta$ for any fixed matrix $B$, which we obtain by Proposition \ref{prop:approxmatrixprod} by having $O(d/(\delta\eps^2))$ non-zeros on the diagonal of $D$). By the normal equations, we have $A^\top  (Ax^* - b) = 0$, thus  $ \langle A (\hat{x} - x^*), (Ax^* - b)\rangle = 0$, and so by the Pythagorean theorem we have 
	\begin{align*}
	\|A \hat{x} - b\|_2^2 =\|Ax^* - b\|_2^2 + \|A(\hat{x} - x^*)\|_2^2  
	\leq (1+4\eps)\|Ax^* - b\|_2^2 
	\end{align*}
	Which completes the proof after rescaling of $\eps$. The runtime required to obtain the matrix $D$ is  $\tilde{O}( \sum_{i=1}^q \nnz(A_i) + \poly(dq/\eps))$ by Proposition \ref{prop:leveragesamplemain}, where we set $D$ to have $m = \Theta(d/(\delta\eps^2))$ non-zero entries on the diagonal. Once $D$ is obtained, one can compute $D(A+b)$ in time $O(md)$, thus the required time is $O(\delta^{-1}(d/\eps)^2)$. Finally, computing $\hat{x}$ once $DA,Db$ are computed requires a single pseudo-inverse computation, which can be carried out in $O(\delta^{-1}d^3/\eps^2)$ time (since $DA$ now has only $O(\delta^{-1}(d/\eps)^2)$ rows). 

\end{proof}

\subsection{Kronecker Product $\ell_p$ Regression}\label{sec:lpreg}
We now consider $\ell_p$ regression for $1 \leq p < 2$. Our algorithm is stated formally in Algorithm ~\ref{alg:l1}. 
	Our high level approach follows that of \cite{ddhkm09}. Namely, we first obtain a vector $x'$ which is an $O(1)$-approximate solution to the optimal solution. This is done by first constructing (implicitly) a matrix $U \in \R^{n \times d}$ that is a well-conditioned basis for the design matrix $A_1 \otimes \cdots \otimes A_q$. We then efficiently sample rows of $U$ with probability proportional to their $\ell_p$ norm (which must be done without even explicitly computing most of $U$). We then use the results of $\cite{ddhkm09}$ to demonstrate that solving the regression problem constrained to these sampled rows gives a solution $x' \in \R^d$ such that $\|(A_1 \otimes \cdots \otimes A_q)x' - b\|_p \leq 8 \min_{x \in \R^d}\|(A_1 \otimes \cdots \otimes A_q)x' - b\|_p$.
	
	We define the \textit{residual error} $\rho = (A_1 \otimes \cdots \otimes A_q)x' - b \in \R^n$ of $x'$. Our goal is to sample additional rows $i \in [n]$ with probability proportional to their residual error $|\rho_i|^p/\|\rho\|_p^p$, and solve the regression problem restricted to the sampled rows. However, we cannot afford to compute even a small fraction of the entries in $\rho$ (even when $b$ is dense, and certainly not when $b$ is sparse). So to carry out this sampling efficiently, we design an involved, multi-part sketching and sampling routine (described in Section \ref{sec:sampleresidual}). This sampling technique is the main technical contribution of this section, and relies on a number of techniques, such as the Dyadic trick for quickly finding heavy hitters from the streaming literature, and a careful pre-processing step to avoid a $\poly(d)$-blow up in the runtime. Given these samples, we can obtain the solution $\hat{x}$ after solving the regression problem on the sampled rows, and the fact that this gives a $(1 + \eps)$ approximate solution will follow from Theorem $6$ of \cite{ddhkm09}.

	\begin{algorithm*}[!h]\caption{Our $\ell_p$ Kronecker Product Regression Algorithm, $1 \leq p < 2$}\label{alg:l1}
		\begin{algorithmic}[1]
			\Procedure{$O(1)$-approximate $\ell_p$ Regression}{$\{ A_i, n_i, d_i \}_{i \in [q] }$} \Comment{Theorem~\ref{thm:l1_regression}}
			
			\State $d \leftarrow \prod_{i=1}^q d_i$, $n \leftarrow \prod_{i=1}^q n_i$.
			\For{$i = 1 ,\dots, q$}
			\State $s_i \leftarrow O( q d_i^2)$
			\State Generate sparse $p$-stable transform $S_i \in \R^{s_i \times n}$ (def \ref{def:sparse_sketch}) \Comment{Lemma ~\ref{lem:lowdistembed}}
			\State Take the QR factorization of $S_i A_i = Q_i R_i$ to obtain $R_i \in \R^{d_i \times d_i}$ \Comment{Fact ~\ref{fact:wellcond}}
			\State Let $Z \in \R^{d \times \tau}$ be a dense $p$-stable transform for $\tau = \Theta(\log(n))$
			\Comment{Definition \ref{def:dense_sketch}}
			\For{$j = 1,\dots,n_i$}
			
			\State $a_{i,j} \leftarrow   \median_{\eta \in[\tau]}\{ (|(A_iR_i^{-1} Z)_{j,\eta}|/ \theta_p	)^p		\}$, where $\theta_p$ is the median of $\mathcal{D}_p$.
			\EndFor
			\EndFor
			\State Define a distribution $\mathcal{D} = \{q_1',q_1',\dots,q_n'\}$ by $q_{{\sum_{i=1}^q j_i \prod_{l = 1}^{j-1}  n_l }}' = \prod_{i=1}^q a_{i,j_i}$.
			\State Let $\Pi \in \R^{n \times n}$ denote a diagonal sampling matrix, where $\Pi_{i,i} = 1/q_i^{1/p}$ with probability $q_i =\min\{1, r_1 q_i'\}$ and $0$ otherwise, where $r_1 = \Theta(d^3/\eps^2)$.  \Comment{\cite{ddhkm09}}
			\State Let $x' \in \R^{d}$ denote the solution of \\
			\hspace{30mm} $\min_{x \in \R^{d} } \| \Pi ( A_1  \otimes A_2 \otimes \cdots \otimes A_q  ) x - \Pi b \|_p$
			\State \Return $x'$ \Comment{$x'$ is an $O(1)$ approx: Lemma \ref{lem:O1approx}}
			\EndProcedure
			\Procedure{$(1+\eps)$-approximate $\ell_p$ Regression}{$x'\in \R^d$}
			
			\State Implicitly define $\rho =( A_1  \otimes A_2 \otimes \cdots \otimes A_q  ) x' -  b \in \R^{n}$
			\State Via Lemma~\ref{lem:lpsampling}, compute a diagonal sampling matrix  $\Sigma \in \R^{n \times n}$ such that $\Sigma_{i,i} = 1/\alpha_i^{1/p}$ with probability $\alpha_i = \min\{	1, \max\{q_i ,r_2 |\rho_i|^p/\|\rho\|_p^p\}	\}$ where $r_2 = \Theta(d^3/\eps^3)$. 
			\State Compute  $\wh{x} =\arg \min_{x \in \R^d} \|\Sigma(A_1 \otimes A_2 \otimes \cdots \otimes A_q) - \Sigma b\|_p$ (via convex optimization methods, e.g., \cite{bcll18,akps19,lsz19})
			\State \Return $\wh{x}$
			\EndProcedure
			
		\end{algorithmic}
	\end{algorithm*}

	\subsubsection{The $\ell_p$ Regression Algorithm}\label{app:lpReg}
	We now give a complete proof of Theorem \ref{thm:l1_regression}.
	Our high level approach follows that of \cite{ddhkm09}. Namely, we first obtain a vector $x'$ which is a $O(1)$ approximate solution to the optimal, and then use the \textit{residual error} $\rho \in \R^d$ of $x'$ to refine $x'$ to a $(1 \pm \eps)$ approximation $\wh{x}$. The fact that $x'$ is a constant factor approximation follows from our Lemma \ref{lem:O1approx}. Given $x'$, by Lemma \ref{lem:lpsampling} we can efficiently compute the matrix $\Sigma$ which samples from the coordinates of the residual error $\rho = (A_1 \otimes \cdots \otimes A_q)x' - b$ in the desired runtime. The sampling lemma is the main technical lemma, and requires a careful multi-part sketching and sampling routine. Given this $\Sigma$, the fact that $\hat{x}$ is a $(1 + \eps)$ approximate solution follows directly from Theorem $6$ of \cite{ddhkm09}. Our main theorem and its proof is stated below. The proof will utilize the lemmas and sampling algorithm developed in the secitons which follow. \\	
\begin{theorem}[Main result, $\ell_p$ $(1+\epsilon)$-approximate regression]\label{thm:l1_regression}
	Fix $1 \leq p < 2$. Then for any constant $q = O(1)$, given matrices $A_1, A_2, \cdots, A_q$, where $A_i \in \R^{n_i \times d_i}$, let $n = \prod_{i=1}^q n_i$, $d = \prod_{i=1}^q d_i$.
	Let $\hat{x} \in \R^d$ be the output of Algorithm \ref{alg:l1}. Then 
	
	\[\|  (A_1 \otimes A_2 \otimes \cdots \otimes A_q) \hat{x} - b \|_p  \leq (1+ \eps) \min_{x \in \R^n}\|  (A_1 \otimes A_2 \otimes \cdots \otimes A_q) x -b \|_p	\]
	holds with probability at least $1-\delta$. In addition, our algorithm takes 
	\begin{align*}
	\wt{O} \left( \left( \sum_{i=1}^q \nnz(A_i)  + \nnz(b) + \poly(d\log(1/\delta)/\eps) \right) \log(1/\delta) \right)
	\end{align*}
	time to output $\wh{x} \in \R^d$.   
\end{theorem}
	
	\begin{proof}
		By Lemma \ref{lem:O1approx}, the output $x'$ in line 16 of Algorithm \ref{alg:l1} is an $8$ approximation of the optimal solution, and $x'$ is obtained in time $\wt{O}(\sum_{i=1}^q \nnz(A_i) + (dq/\eps)^{O(1)})$. We then obtain the residual error $\rho = (A_1 \otimes \cdots \otimes A_q)x' - b$ (implicitly). By Theorem 6 of \cite{ddhkm09}, if we let $\Sigma \in \R^{n \times n}$ be a row sampling matrix where $\Sigma_{i,i} = 1/\alpha_i^{1/p}$ with probability $\alpha_i = \min\{	1, \max\{	q_i, r_2 \frac{|\rho_i|^p}{\|\rho\|_p^p}	\}$, where $q_i$ is the row sampling probability used in the sketch $\Pi$ from which $x'$ was obtained, and $r_2 = O(d^3 /\eps^2\log(1/\eps))$, then the solution to $\min_x \|\Sigma (A_1 \otimes \cdots \otimes A_q)x - \Sigma b\|_p$ will be a $(1 + \eps)$ approximately optimal solution. By Lemma \ref{lem:lpsampling}, we can obtain such a matrix $\Sigma$ in time  $\wt{O}(\sum_{i=1}^q \nnz(A_i) + q\nnz(b) + (d\log(n)/(\eps\delta )^{O(q^2)})$, which completes the proof of correctness. Finally, note that we can solve the sketched regression problem $\min_x \|\Sigma (A_1 \otimes \cdots \otimes A_q)x - \Sigma b\|_p$ which has $O((d\log(n)/\eps)^{O(q^2)} (1/\delta))$ constraints and $d$ variables in time $O((d\log(n)/\eps)^{O(q^2)} (1/\delta))$ using linear programming for $p=1$ (see \cite{ls14,ls15,cls19,lsz19,b20} for some state of the art linear program solvers), or more generally interior point methods for convex programming for $p> 1$ (see \cite{bcll18,akps19,lsz19} for the recent development of $\ell_p$ solver).
		
		Now to boost the failure probability from a $O(1/\delta)$ to $\log(1/\delta)$ dependency, we do the following. We run the above algorithm with $\delta = 1/10$, so that our output $\hat{x} \in \R^d$ is a $(1+\eps)$ approximation with probability $9/10$, and we repeat this $r$ times with $r=O(\log(1/\delta))$ time to obtain $\hat{x}_1 , \hat{x}_2 ,\dots, \hat{x}_r,$, and then we repeat another $r$ times to obtain distinct sampling matrices $\Sigma_1,\dots,\Sigma_r$ (note that $\Sigma_i$ is not the sampling matrix associated to $\hat{x}_i$ in this notation, and comes from a distinct repetition of the above algorithm). This blows up the overall runtime by $O(\log(1/\delta))$. 
		Now for any vector $x \in \R^d$, let $X_i = |\Sigma_{i,i}(A_1 \otimes \cdots \otimes A_q)_{i,*}x -  b_i )|^p$. Clearly $\ex{\sum_i X_i} = \| (A_1 \otimes \cdots \otimes A_q)x -  b \|_p^p$. Moreover, we can bound $\ex{\sum_i X_i^2}$ by $\poly(d) \left(\ex{\sum_i X_i}\right)^2 / r_2$ (see proof of Lemma 9 in \cite{ddhkm09} for a computation). Setting $r_2 = \poly(d)$ large enough, by Chebyshev's we have that each $\Sigma_i$ preserves the cost of a fixed vector $x_j$ with probability $99/100$. 		
		so with probability $1-\delta$, after a union bound, for all $i \in [r]$ we have that 
		\[	\median_j \|\Sigma_j  (A_1 \otimes \cdots \otimes A_q)\hat{x}_i - \Sigma_j b\|_p\ = (1 \pm \eps ) 	\| (A_1 \otimes \cdots \otimes A_q)\hat{x}_i - b\|_p \]
		
		Thus we now have $(1+\eps)$-error approximations of the cost of each $\hat{x}_i$, and we can output the $\hat{x}_i$ with minimal cost. By Chernoff bounds, at least one $\hat{x}_i$ will be a $(1+\eps)$ optimal solution, so by a union bound we obtain the desired result with probability $1-2\delta$ as needed.

	\end{proof}

	We start by defining a tensor operation which will be useful for our analysis.  
	
	\begin{definition}[ $((\cdot , \dots ,\cdot ),\cdot)$ operator for tensors and matrices]\label{def:bracket}
		Given tensor $A\in \mathbb{R}^{d_1 \times d_2 \times \dots \times d_q}$ and matrices $B_i\in \mathbb{R}^{n_i\times d_i}$ for $i \in [q]$, we define the tensor $((B_1,B_2,\dots,B_q),A)\in \mathbb{R}^{n_1\times n_2 \times \dots \times n_q}$: 
		\[((B_1,B_2,\dots,B_q),A)_{i_1,\dots,i_q} =  \sum_{i_1'=1}^{d_1} \sum_{i_2'=1}^{d_2} \cdots \sum_{i_q'=1}^{d_q}A_{i_1',i_2',\dots,i_q'} \prod_{\ell=1}^q (B_\ell)_{i_\ell,i_\ell'} \]
		
		Observe for the case of $q=2$, we just have $((B_1,B_2), A) = B_1 A B_2^\top \in \R^{n_1 \times n_2}$. 
	\end{definition}
	Using the above notation, we first prove a result about reshaping tensors.

	\begin{lemma}[Reshaping]\label{lem:reshaping} 
		Given matrices $A_1, A_2, \cdots, A_q \in \R^{n_i \times d_i}$ and a tensor $B \in \R^{n_1 \times n_2 \times \cdots \times n_q}$, let $n = \prod_{i=1}^q n_i$ and let $d = \prod_{i=1}^d d_i$. Let $b$ denote the vectorization of $B$.
		For any tensor $X \in \R^{d_1 \times d_2 \times \cdots \times d_q}$, we have
		$\| ((A_1,A_2, \cdots, A_q),X) - B \|_{\xi}$ is equal to
		$ \| (A_1 \otimes A_2 \otimes \cdots \otimes A_q) x - b \|_{\xi}$
		where $\xi$ is any entry-wise norm (such as an $\ell_p$-norm) and $x$ is the vectorization of $X$. See Definition \ref{def:bracket} of the $((\cdot, \dots ,\cdot), \cdot)$ tensor operator. 
		
		Observe, for the case of $q=2$, this is equivalent to the statement that $\|A_1 X A_2^\top  - B\|_\xi = \|(A_1 \otimes A_2)x - b\|_\xi$.
		
	\end{lemma}
	\begin{proof}
		For the pair $x \in \R^{d}$, $X \in \R^{d_1 \times d_2 \times \cdots \times d_q}$, the connection is the following: $\forall i_1 \in [d_1],  \dots, i_q \in [d_q]$,
		\begin{align*}
		x_{ i_1 + \sum_{l=2}^q ( i_l - 1) \cdot \prod_{t=1}^{l-1} d_t  } = X_{i_1, \cdots, i_q}.
		\end{align*}
		Similarly, for $b \in \R^{n}$, $B \in \R^{n_1 \times n_2 \times \cdots \times n_q}$,  for any $j_1, \in [n_1], \dots, j_q \in [n_q]$,
		\begin{align*}
		b_{ j_1 + \sum_{l=2}^q ( j_l - 1) \cdot \prod_{t=1}^{l-1} n_t  } = B_{j_1, j_2, \cdots, j_q}.
		\end{align*}
		
		For simplicity, for any $(i_1,\dots,i_q) \in [d_1] \times \dots \times [d_q]$ and $(j_1,\dots,j_q) \in [n_1]\times \dots \times [n_q]$ we define
		$\vec{i} =  i_1 + \sum_{l=2}^q ( i_l - 1) \cdot  \prod_{t=1}^{l-1} d_t$ and similarly $\vec{j} = j_1 + \sum_{l=2}^q ( j_l - 1) \cdot  \prod_{t=1}^{l-1} n_t $. 
		Then we can simplify the above relation and write 
		$x_{ \vec{i} } = X_{i_1,i_2,\cdots,i_q}, \text{~and~} b_{ \vec{j} } = B_{j_1,j_2, \cdots, j_q}.$
		
		For a matrix $Z$, let $Z_{i,*}$ denote the $i$-th row of $Z$. We consider the $\vec{j}$-th entry of $(A_1 \otimes A_2 \otimes \cdots \otimes A_q )x$, 
		\begin{align*}
		((A_1 \otimes A_2 \otimes \cdots \otimes A_q )x)_{\vec{j}}  =& ~ \left\langle  \left( A_1 \otimes A_2 \otimes \cdots \otimes A_q \right)_{\vec{j},*} \cdot x \right\rangle \\
		= & ~ \sum_{i_1 = 1}^{d_1} \sum_{i_2 = 1}^{d_2} \cdots \sum_{i_q = 1}^{d_q} \left( \prod_{l=1}^q (A_l)_{j_l, i_l} \right) \cdot  x_{ \vec{i} } \\
		= & ~ \sum_{i_1 = 1}^{d_1} \sum_{i_2 = 1}^{d_2} \cdots \sum_{i_q = 1}^{d_q}  \left( \prod_{l=1}^q (A_l)_{j_l, i_l} \right) \cdot X_{i_1,i_2, \cdots, i_q} \\
		= & ~ (( A_1, A_2, \cdots, A_q ),X )_{j_1,\dots,j_q}.
		\end{align*}Where the last equality is by Definition (\ref{def:bracket}). 
		Since we also have $b_{\vec{j}} =B_{j_1,\dots,j_q} $, this completes the proof of the Lemma.
	\end{proof}

	\subsubsection{Sampling From an $\ell_p$-Well-Conditioned Base}\label{sec:l1_tensor_well_conditioned_basis}

	In this Section, we discuss the first half of Algorithm \ref{alg:l1} which computes $x' \in \R^d$, which we will show is a $O(1)$-approximate solution to the optimal. First note that by Lemma \ref{lem:lowdistembed} together with fact \ref{fact:wellcond}, we know that $A_i R_i^{-1}$ is an $\ell_p$ well conditioned basis for $A_i$ (recall this means that $A_i R_i^{-1}$ is a $(\alpha,\beta,p)$ well conditioned basis for $A$, and $\beta/\alpha = d_i^{O(1)}$) with probability $1-O(1/q)$, and we can then union bound over this occurring for all $i \in [q]$. Given this, we now prove that  $(A_1 R_1^{-1} \otimes A_2 R_2^{-1} \otimes \dots \otimes A_q R_q^{-1})$ is a well conditioned basis for $ (A_1\otimes A_2  \otimes \dots \otimes A_q)$.
	\begin{lemma}
		Let $A_i \in \R^{n_i \times d_i}$ and $R_i \in \R^{d_i \times d_i}$. Then if $A_i R_i^{-1}$ is a $(\alpha_i,\beta_i,p)$ well-conditioned basis for $A_i$ for $i =1,2,\dots,q$, we have for all $x \in \R^{d_1\cdots d_q}$:
		\begin{align*}
		\prod_{i=1}^q \alpha_i \| x \|_p \leq \| (A_1 R_1^{-1} \otimes A_2 R_2^{-1} \otimes \dots \otimes A_q R_q^{-1}) x \|_p \leq\prod_{i=1}^q \beta_i \| x \|_p
		\end{align*}
	\end{lemma}
	\begin{proof}
		We first consider the case of $q=2$. We would like to prove
		\begin{align*}
		\alpha_1 \alpha_2 \| x \|_p \leq \| (A_1 R_1^{-1} \otimes A_2 R_2^{-1}) x \|_p \leq \beta_1 \beta_2 \| x \|_p,
		\end{align*}
		First note, by the reshaping Lemma \ref{lem:reshaping}, this is equivalent to 
		\begin{align*}
		\alpha_1 \alpha_2 \| X \|_p \leq  \| A_1 R_1^{-1} X  ( R_2^{-1}  A_2 )^\top \|_p \leq \beta_1 \beta_2 \| X \|_p .
		\end{align*}
		Where $X \in \R^{d_1 \times d_2}$ is the tensorization of $x$. We first prove one direction. Let $U_1 = A_1 R_1^{-1}$ and $U_2 = A_2 R_2^{-1}$. We have
		\begin{align*}
		\| U_1 X U_2^\top \|_p^p = & ~ \sum_{i_2=1}^{n_2} \| U_1 (X U_2^\top)_{i_2} \|_p^p \\
		\leq & ~ \sum_{i_2=1}^{n_2} \beta_1^p \| (X U_2^\top)_{i_2} \|_p^p \\
		= & ~ \beta_1^p \| X U_2^\top \|_p^p \\
		\leq & ~ \beta_1^p \beta_2^p \| X \|_p^p,
		\end{align*}
		where the first step follows from rearranging, the second step follows from the well-conditioned property of $U_1$, the third step follows from rearranging again, the last step follows from  the well-conditioned property of $U_2$.
		Similarly, we have
		\begin{align*}
		\| U_1 X U_2^\top \|_p^p = & ~ \sum_{i_2=1}^{n_2} \| U_1 (X U_2^\top)_{i_2} \|_p^p \\
		\geq & ~ \sum_{i_2=1}^{n_2} \alpha_1^p \| (X U_2^\top)_{i_2} \|_p^p \\
		= & ~ \alpha_1^p \| X U_2^\top \|_p^p \\
		\geq & ~ \alpha_1^p \alpha_2^p \| X \|_p^p,
		\end{align*}
		where again the first step follows from rearranging, the second step follows from  the well-conditioned property of $U_1$, the third step follows from rearranging again, the last step follows from  the well-conditioned property of $U_2$.
		
		In general, for arbitrary $q \geq 2$, similarly using our reshaping lemma, we have 
		\begin{align*}
		\| (\otimes_{i=1}^q (A_i R_i^{-1}) )x \|_p \geq & ~ \prod_{i=1}^q \alpha_i \| x \|_p , \\
		\| (\otimes_{i=1}^q (A_i R_i^{-1})) x \|_p \leq & ~ \prod_{i=1}^q \beta_i \| x \|_p .
		\end{align*}
	\end{proof}
	
	Putting this together with fact \ref{fact:wellcond}, and noting $d = d_1 \cdots d_q$, we have
	\begin{corollary}\label{cor:wellcond}
		Let $A_i R_i^{-1}$be as in algorithm \ref{alg:l1}. Then we have for all $x \in \R^{d_1 \cdots d_q}$:
		\begin{align*}
		(1/d)^{O(1)}\| x\|_p \leq \| (A_1 R_1^{-1} \otimes \cdots \otimes A_q R_q^{-1}) x \|_p \leq d^{O(1)}\| x \|_p,
		\end{align*}
		In other words, $ (A_1 R_1^{-1} \otimes \cdots \otimes A_q R_q^{-1})$ is a well conditioned $\ell_p$ basis for $ (A_1\otimes \cdots \otimes A_q )$
	\end{corollary}
	
	From this, we can obtain the following result.
	
	\begin{lemma}\label{lem:O1approx}
		Let $x' \in \R^d$ be the output of the $O(1)$-Approximate $\ell_p$ Regression Procedure in Algorithm \ref{alg:l1}. Then with probability $99/100$ we have 
		\begin{align*}
		\|(A_1 \otimes \cdots \otimes A_q) x' - b\|_p \leq 8\min_x \|(A_1 \otimes \cdots \otimes A_q) x - b\|_p .
		\end{align*}
		 Moreover, the time required to compute $x'$ is $\tilde{O}(\sum_{i=1}^q \nnz(A_i) + \poly(d q/\eps) )$. 
	\end{lemma}
	\begin{proof}
		By Theorem $6$ of \cite{ddhkm09}, if we let $\Pi$ be a diagonal row sampling matrix such that $\Pi_{i,i} = 1/q_i^{1/p}$ with probability $q_i \geq \min\{1, r_1 \frac{\|U_{i,*}\|_p^p}{\|U\|_p^p}\}$, where $U$ is a $\ell_p$ well-conditioned basis for  $(A_1\otimes \cdots \otimes A_q )$ and $r_1 = O(d^3)$, then the solution $x'$ to 
		\begin{align*}
		\min_x \|\Pi((A_1 \otimes \cdots \otimes A_q)x - b\|
		\end{align*}
		will be a $8$-approximation.
		Note that we can solve the sketched regression problem $\min_x \|\Pi((A_1 \otimes \cdots \otimes A_q)x' - b\|$ which has $O(\poly(d/\eps))$ constraints and $d$ variables in time $\poly(d/\eps)$ using linear programming for $p=1$ (see \cite{ls14,ls15,cls19,lsz19,b20} for the state of the art linear program solver), or more generally interior point methods for convex programming for $p> 1$ (see \cite{bcll18,akps19,lsz19} for the recent development of $\ell_p$ solver). 
		
		Then by Corollary \ref{cor:wellcond}, we know that setting $U =  (A_1 R_1^{-1} \otimes \cdots \otimes A_q R_q^{-1})$ suffices, so now we must sample rows of $ U$. To do this, we must approximately compute the norms of the rows of $U$. Here, we use the fact that $\|\cdot\|_p^p$ norm of a row of $ (A_1 R_1^{-1} \otimes \cdots \otimes A_q R_q^{-1} )$ is the product of the row norms of the $A_i R_i^{-1}$ that correspond to that row. Thus it suffices to sample a row $j_i$ from each of the $A_i R_i^{-1}$'s with probability at least $\min\{1,r_1\|(A_i R_i^{-1})_{j_i,*}\|_p^p/\|A_i R_i^{-1}\|_p^p\}$ for each $i \in [q]$. 
		
		To do this, we must estimate all the row norms $\|(A_i R_i^{-1})_{j_i,*}\|_p^p$ to $(1 \pm 1/10)$ error. This is done in steps $7-10$ of Algorithm \ref{alg:l1}, which uses dense $p$-stable sketches $Z \in \R^{d \times \tau}$, and computes $(A_i R_i^{-1}Z)$, where $\tau = \Theta(\log(n))$. Note that computing $R_i^{-1}Z \in \R^{d \times \tau}$ requires $\tilde{O}(d^2)$. Once computed,  $A_i (R_i^{-1}Z)$ can be computed in $\wt{O}(\nnz(A_i))$ time. We then take the median of the coordinates of $(A_i R_i^{-1}Z)$ (normalized by the median of the $p$-stable distribution $\mathcal{D}_p$, which can be efficiently approximated to $(1 \pm \eps)$ in $O(\poly(1/\eps))$ time, see Appendix A.2 of \cite{kane2010exact} for details) as our estimates for the row norms. This is simply the Indyk median estimator \cite{i06}, and gives a $(1 \pm 1/10)$ estimate $a_{i,j}$ of all the row norms $\|(A_i R_i^{-1})_{j,*}\|_p^p$  with probability $1-1/\poly(n)$. Then it follows by Theorem 6 of \cite{ddhkm09} that $x'$ is a $8$-approximation of the optimal solution with probability $99/100$ (note that we amplified the probability by increasing the sketch sizes $S_i$ by a constant factor), which completes the proof.

	\end{proof}

	\subsubsection{$\ell_p$ Sampling From the Residual of a $O(1)$-factor Approximation}\label{sec:sampleresidual}
	By Lemma \ref{lem:O1approx} in the prior section, we know that the $x'$ first returned by the in algorithm \ref{alg:l1} is a $8$-approximation. We now demonstrate how we can use this $O(1)$ approximation to obtain a $(1+\eps)$ approximation. The approach is again to sample rows of $(A_1 \otimes \cdots \otimes A_q)$. But instead of sampling rows with the well-conditioned leverage scores $q_i$, we now sample the $i$-th row with probability $\alpha_i = \min\{	1, \max\{q_i ,r_2 |\rho_i|^p/\|\rho\|_p^p\}	\}$
	, where $\rho = (A_1 \otimes \cdots \otimes A_q)x' - b \in \R^n$ is the \textit{residual error} of the $O(1)$-approximation $x'$. Thus we must now determine how to sample quickly from the residuals $|\rho_i|^p/\|\rho\|_p^p$. Our sampling algorithm will need a tool originally developed in the streaming literature.
	
	\paragraph{Count-sketch for heavy hitters with the Dyadic Trick.}
	We now introduce a sketch $S$ which finds the $\ell_2$ heavy hitters in a vector $x$ efficently. This sketch $S$ is known as count-sketch for heavy hitters with the Dyadic Trick. To build $S$ we first stack $\Theta(\log(n))$ copies of the \textit{count sketch matrix} $S^i \in \R^{k' \times n}$ \cite{cw13}. The matrix $S^i$ is constructed as follows. $S^i$ has exactly one non-zero entry per column, which is placed in a uniformly random row, and given the value $1$ or $-1$ uniformly at random. For $S^i$, let $h_i:[n] \to [k']$ be such that $h_i(t)$ is the row with the non-zero entry in the $t$-th column of $S^i$, and let $g_i:[n] \to \{1,-1\}$ be such that the value of that non-zero entry is $g_i(t)$. Note that the $h_i,g_i$ can be implemented as $4$-wise independent hash functions. Fix any $x \in \R^n$. Then given $S^1x,S^2x,\cdots ,S^{\Theta(\log(n))}x$, we can estimate the value of any coordinate $x_j$ by $\median_{i \in \Theta{\log(n)}} \{	g_i(j)(S^ix)_{h_i(j) }\}$.
	
	It is well-known that this gives an estimate of $x_j$ with additive error $\Theta(1/\sqrt{k'})\|x\|_2$ with probability $1-1/\poly(n)$ for all $j \in [n]$ \cite{charikar2004finding}. However, naively, to find the heaviest coordinates in $x$, that is all coordinates $x_j$ with $|x_j| \geq \Theta(1/\sqrt{k'})\|x\|_2$, one would need to query $O(n)$ estimates. This is where the Dyadic trick comes in \cite{cormode2005improved}. We repeat the above process $\Theta(\log(n))$ times, with matrices $S^{(i,j)}$, for $i,j \in \Theta(\log(n))$. Importantly, however, in $S^{(i,j)}$, for all $t,t' \in [n]$ such that the first $j$ most significant bits in their binary identity representation are the same, we set $h_{(i,j)}(t) = h_{(i,j)}(t')$, effectively collapsing these identities to one. To find a heavy item, we can then query the values of the *two* identities from $S^{(1,1)},S^{(2,1)},\cdots,S^{(\Theta(\log(n)) , 1)}$, and recurse into all the portions which have size at least $\Theta(1/\sqrt{k'})\|x\|_2$. It is easy to see that we recurse into at most $O(k')$ such pieces in each of the $\Theta(\log(n))$ levels, and it takes $O(\log(n))$ time to query a single estimate, from which the desired runtime of $O(k' \log^2(n))$ is obtained. For a further improvement on size $k$ of the overall sketched required to quickly compute $Q$, see \cite{larsen2016heavy}. We summarize this construction below in definition \ref{def:csHH}.
	
	\begin{definition}[Count-sketch for heavy hitters with Dyadic Trick \cite{charikar2004finding, larsen2016heavy}]\label{def:csHH}
		There is a randomized sketch $S \in \R^{k \times n}$ with $k = O(\log^2(n)/\eps^2)$ such that, for a fixed vector $x \in \R^{n}$, given $Sx \in \R^{k}$, one can compute a set $Q \subset [n]$ with $|Q| = O(1/\eps^2)$ such that $\{ i \in [n] \;| \; |x_i|\geq \eps \|x\|_2 \} \subseteq Q$ with probability $1-1/\poly(n)$. Moreover, $Sx$ can be computed in $O(\log^2(n)\nnz(x))$ time. Given $Sx$, the set $Q$ can be computed in time $O(k)$. 
	\end{definition}

	We begin with some notation. For a vector $y \in \R^{n}$, where $n = n_1 \cdots n_q$, one can index any entry of $y_i$ via $\vec{i} = (i_1,i_2,\cdots,i_q) \in [n_1] \times \cdots \times [n_q]$ via $i = i_1 + \sum_{j=2}^q (i_j-1) \prod_{l = 1}^{i_j-1} n_l$. It will useful to index into such a vector $y$ interchangably via a vector $y_{\vec{i}}$ and an index $y_{j}$ with $j \in [n]$.
	For any set of subsets $T_i \subset [n_i]$, we can define $y_{T_1 \times \cdots T_q} \in \R^{n}$ as $y$ restricted to the $\vec{i} \in T_1 \times \cdots \times T_q$. Here, by restricted, we mean the coordinates in $y$ that are not in this set are set equal to $0$. Similarly, for a $y \in \R^{n_i}$ and $S \subset [n_i]$, we can define $y_{S}$ as $y$ restricted to the coordinates in $S$. Note that in Algorithm \ref{alg:residualsample}, $\mathbb{I}_n$ denotes the $n \times n$ identity matrix for any integer $n$. We first prove a proposition on the behavior of Kronecker products of $p$-stable vectors, which we will need in our analysis.

	\begin{proposition}\label{prop:stablekronecker}
		Let $Z_1,Z_2,\cdots,Z_q$ be independent vectors with entries drawn i.i.d. from the $p$-stable distribution, with $Z_i \in \R^{n_i}$. Now fix any  $i \in [q]$, and any $x \in \R^{n}$, where $n = n_1 n_2 \cdots n_q$. Let $e_j \in \R^{n_i}$ be the $j$-th standard basis column vector for any $j \in [n_i]$. Let $\Gamma(i,j)  = [n_1] \times [n_2] \times \cdots \times [n_{i-1}] \times \{j\} \times [n_{i+1}] \times \cdots \times [n_q]$. Define the random variable 
		\begin{align*}
		\mathcal{X}_{i,j}(x) = |(Z_1 \otimes Z_1 \otimes \cdots \otimes Z_{i-1} \otimes e_j^\top \otimes Z_{i+1} \otimes \cdots \otimes Z_q)x|^p .
		\end{align*}
		Then for any $\lambda > 1$, with probability at least $1-O(q/\lambda)$ we have \[  \|x_{\Gamma(i,j) }\|_p^p /\lambda^q \leq\mathcal{X}_{i,j}(x) \leq  (\lambda \log(n))^q\|x_{\Gamma(i,j) }\|_p^p \]
	\end{proposition} 
	\begin{proof}
		First observe that we can reshape $y = x_{\Gamma} \in \R^m$ where $m = n/n_i$, and re-write this random variable as $\mathcal{X}_{i,j}(x) = |(Z_1 \otimes Z_2 \otimes \cdots \otimes Z_{q-1}) y|^p$. By reshaping Lemma \ref{lem:reshaping}, we can write this as $|(Z_1 \otimes Z_2 \otimes \cdots \otimes Z_{q-2}) YZ_{q-1}^\top|^p$, where $Y \in \R^{m/n_{q-1} \times n_{q-1}}$. We first prove a claim. In the following, for a matrix $A$, let $\|A\|_p^p = \sum_{i,j} |A_{i,j}|^p$. 
		
		\begin{claim}
			Let $Z$ be any $p$-stable vector and $X$ a matrix. Then for any $\lambda > 1$, with probability $1-O(1/\lambda)$, we have 
			\begin{align*}
			\lambda^{-1} \|X\|_p^p\leq \|XZ\|_p^p \leq \log(n) \lambda \|X\|_p^p .	
			\end{align*}	
		\end{claim}
		\begin{proof}
			By $p$-stability, each entry of $|(XZ)_i|^p$ is distributed as $|z_i|^p\|X_{i,*}\|_p^p$, where $z_i$ is again $p$-stable (but the $z_i's$ are not independent). Now $p$-stables have tails that decay at the rate $\Theta(1/x^p)$ (see Chapter 1.5 of \cite{nolan2009stable}), thus $\pr{|z_i|^p > x} = O(1/x)$ for any $x > 0$. We can condition on the fact that $z_i < \lambda \cdot n^{10}$ for all $i$, which occurs with probability at least $1-n^{-9}/\lambda$ by a union bound. Conditioned on this, we have $\ex{|z_i|^p} = O(\log(n))$ (this can be seen by integrating over the truncated tail $O(1/x)$), and the upper bound then follows from a application of Markov's inequality. 
			
			For the lower bound Let $Y_i$ be an indicator random variable indicating the event that $|z_i|^p < 2/\lambda$. Now $p$-stables are anti-concentrated, namely, their pdf is upper bounded by a constant everywhere. It follows that $\pr{Y_i } < c/\lambda$ for some constant $c$. By Markov's inequality $\pr{ \sum_i Y_i \|X_{i,*}\|_p^p > \|X\|_p^p/2 } < O(1/\lambda)$. Conditioned on this, the remaining $\|X\|_p^p/2$ of the $\ell_p$ mass shrinks by less than a $2/\lambda$ factor, thus $\|XZ\|_p^p > (\|X\|_p^p/2)(2/\lambda) = \|X\|_p^p/\lambda$ as needed. 
		\end{proof}
		By the above claim, we have $ \|Y\|_p/\lambda^{1/p} \leq \|YZ_{q-1}^\top\|_p \leq  (\log(n) \lambda)^{1/p} \|Y\|_p$ with probability $1-O(1/\lambda)$. Given this, we have $\mathcal{X}_{i,j}(x) = |(Z_1 \otimes Z_2 \otimes \cdots \otimes Z_{q-2}) y'|^p$, where  $ \|Y\|_p/\lambda^{1/p} \leq \|y'\|_p  \leq  (\log(n) \lambda)^{1/p} \|Y\|_p$. We can inductively apply the above argument, each time getting a blow up of  $(\log(n) \lambda)^{1/p}$ in the upper bound and $(1/\lambda)^p$ in the lower bound, and a failure probability of $(1/\lambda)$. Union bounding over all $q$ steps of the induction, the proposition follows. 
		
	\end{proof}

	\begin{algorithm*}\caption{Algorithm to $\ell_p$ sample $\Theta(r_2)$ entires of $\rho = (A_1 \otimes \cdots \otimes A_q)x'-b$}\label{alg:residualsample}
		\begin{algorithmic}[1]
			\Procedure{Residual $\ell_p$ sample}{$\rho, r_2$} 
			\State $r_3 \leftarrow \Theta(r_2 \log^{q^2}(n)/\delta)$.
			\State Generate i.i.d. $p$-stable vectors $Z^{1,j}, Z^{2,j},\dots, Z^{q,j} \in \R^n$ for $j \in [\tau]$ for $\tau = \Theta(\log(n))$
			
			\State $T \leftarrow \emptyset$ \Comment{sample set to return}
			\State Pre-compute  and store $Z^{i,j}A_i \in \R^{1 \times d_i}$ for all $i \in [q]$ and $j \in [\tau]$
			\State Generate count-sketches for heavy hitters $S^i \in \R^{k \times n_i}$ of Definition \ref{def:csHH} for all $i \in [q]$, where $k = O(\log^2(n) r_3^{O(1)})$.
			\For{$t=1,2,\dots,r_3$}
			\State $s = (s_1,\dots,s_q) \leftarrow (\emptyset,\dots,\emptyset)$ \Comment{next sample to return}
			
			\State $w^j \leftarrow \left( (\mathbb{I}_{n_1}) \otimes  (\bigotimes_{k = 2}^{q} Z^{k,j}) \rho\right) \in \R^{n_1}$ \Comment{$\mathbb{I}_n \in \R^{n \times n}$ is identity}
			\State Define $w \in \R^{n_1}$ by $w_l = \median_{j \in [\tau]} \{|w^j_l|\}$ for $l \in [n_1]$
			
			\State Sample $j^* \in [n_1]$ from the distribution $\left(\frac{|w_1|^p}{\|w\|_p^p}, \frac{|w_2|^p}{\|w\|_p^p}, \dots, \frac{|w_{n_1}|^p}{\|w\|_p^p}\right)$
			\State $s_1 \leftarrow j^*$
			
			\For{$i=2,\dots,q$}
			
			\For{ $j \in [\tau]$}
			\State Write $e_{a_k}^\top \in \R^{1 \times n_k}$ as the standard basis vector
			\State $v_i^j \leftarrow S^i\left( (\bigotimes_{k = 1}^{i-1} e_{a_k}^\top) \otimes (\mathbb{I}_{n_i}) \otimes  (\bigotimes_{k = i+1}^{q} Z^{k,j}) \rho\right) \in \R^{k}$ 
			\State Compute heavy hitters $H_{i,j} \subset [n_i]$ from $v_i^j$\Comment{Definition \ref{def:csHH}}
			\State $\beta_i^j \leftarrow \left( (\bigotimes_{k = 1}^{i-1} e_{a_k}^\top) \otimes   (\bigotimes_{k = i}^{q} Z^{k,j}) \rho\right) \in \R$ 
			
			\EndFor
			\State Define $\beta_i \in \R^{k'}$ by $\beta_i = \median_{j \in [\tau]} \{		|\beta_i^j|^p\}$
			\State $H_i = \cup_{j=1}^\tau H_{i,j}$ 
			\State $\gamma_i \leftarrow \median_{j \in [\tau]}\left( (\bigotimes_{k = 1}^{i-1} e_{a_k}^\top) \otimes  Z^{i,j}_{[n_i] \setminus H_i} \otimes  (\bigotimes_{k = i+1}^{q} Z^{k,j}) \rho\right) \in \R$ 
			\If{ with probability $1-\gamma_i/\beta_i$}
			
			\State Draw $\xi \in H_i$ with probability \[\frac{ \text{median}_{j \in \tau}\left|\left( (\bigotimes_{k = 1}^{i-1} e_{a_k}^\top) \otimes (e_\xi^\top) \otimes  (\bigotimes_{k = i+1}^{q} Z^{k,j}) \rho\right)\right|^p}{\sum_{\xi' \in H_i} \text{median}_{j \in \tau}\left|\left( (\bigotimes_{k = 1}^{i-1} e_{a_k}^\top) \otimes (e_{\xi'}^\top) \otimes  (\bigotimes_{k = i+1}^{q} Z^{k,j}) \rho\right)\right|^p}\]
			\State $s_i \leftarrow \xi$
			\Else \Comment{$s_i$ was not sampled as a heavy hitter}
			\State Randomly partition $[n_i]$ into $\Omega_1^i,\Omega_2^i,\dots,\Omega_\eta^i$ with $\eta = \Theta(r_3^2)$
			\State Sample $t \sim [\eta]$ uniformly at random
			\For{$j \in \Omega_t \setminus H_i$}
			\State $\theta_j = \median_{l \in [\tau]} \left(| (\bigotimes_{k = 1}^{i-1} e_{a_k}^\top) \otimes (e_j^\top) \otimes  (\bigotimes_{k = i+1}^{q} Z^{k,l}) \rho |^p\right)$
			\EndFor
			\State 
			Sample $s_i \leftarrow j^*$ from the distribution $\{\frac{\theta_j}{\sum_{j' \in \Omega_t \setminus H_i} \theta_{j'} }\}_{j \in \Omega_t \setminus H_i}$
			\EndIf
			\EndFor
			\State  $T \leftarrow S \cup s$ where $s = (s_1,\dots,s_q)$
			\EndFor
			
			\State
			\Return sample set $T$
			
			\EndProcedure
		\end{algorithmic}
	\end{algorithm*}


	\begin{lemma}\label{lem:lpsampling}
		Fix any $r_2 \geq 1$, and suppose that $x' = \min_x \|\Pi( A_1 \otimes \cdots \otimes A_q)x - \Pi b\|_p$ and $\Pi \in \R^{n \times n}$ is a row sampling matrix such that $\Pi_{i,i} = 1/q_i^{1/p}$ with probability $q_i$. Define the residual error $\rho = (A_1 \otimes \cdots \otimes A_q)x' - b \in \R^n$. Then Algorithm \ref{alg:residualsample}, with probability $1-\delta$, succeeds in outputting a row sampling matrix $\Sigma \in \R^{n \times n}$ such that $\Sigma_{i,i} = 1/\alpha_{i}^{1/p}$ with probability $\alpha_i = \min\{	1, \max\{q_i ,r_3 |\rho_i|^p/\|\rho\|_p^p\}	\}$ for some $r_3 \geq r_2$, and otherwise $\Sigma_{i,i} = 0$. The algorithm runs in time 
		\begin{align*}
		\wt{O} \left( \sum_{i=1}^q \nnz(A_i) + q\nnz(b) + (r_2\log(n)/\delta )^{O(q^2)} \right).
		\end{align*}
	\end{lemma}
	\begin{proof}
		The algorithm is given formally in Figure \ref{alg:residualsample}. We analyze the runtime and correctness here. 
		
		\paragraph{Proof of Correctness.} The approach of the sampling algorithm is as follows. Recall that we can index into the coordinates of $\rho \in \R^n$ via $\vec{a} = (a_1,\dots,a_q)$ where $a_i \in [n_i]$. We build the coordinates of $\vec{a}$ one by one. To sample a $\vec{a} \in \prod_{i=1}^q [n_i]$, we can first sample $a_1 \in [n_1]$ from the distribution $\pr{a_1 = j} = \sum_{\vec{u}: u_1 = j} |\rho_{\vec{u}}|^p/ (\sum_{\vec{u}} |\rho_{\vec{u}}|^p)$. Once we fix $a_1$, we can sample $a_2$ from the conditional distribution distribution $\pr{a_2 = j} = \sum_{\vec{u}: u_2 = j, u_1 = a_1} |\rho_{\vec{u}}|^p/ (\sum_{\vec{u} : u_1 = a_1} |\rho_{\vec{u}}|^p)$, and so on. For notation, given a vector $\vec{a} = (a_1,\dots,a_{i-1})$, let $\Delta(\vec{a}) = \{\vec{u} \in [n_1] \times \cdots \times [n_q] \; | \; a_j = y_j \text{ for all } j=1,2,\dots,i-1 \}$. Then in general, when we have sampled $\vec{a} = (a_1,\dots,a_{i-1})$ for some $i \leq q$, we need to sample $a_i \leftarrow j \in [n_k]$ with probability
		\begin{align*}
		\pr{a_i = j} =  \sum_{\vec{u} \in \Delta(\vec{a}) : u_i = j} |\rho_{\vec{u}}|^p/ \left( \sum_{\vec{u} \in \Delta(\vec{a}) } |\rho_{\vec{u}}|^p \right) .
		\end{align*}
		We repeat this process to obtain the desired samples.  Note that to sample efficiently, we will have to compute these aforementioned sampling probabilities approximately. Because of the error in approximating, instead of returning $r_2$ samples, we over-sample and return  $r_3 = \Theta(r_2 \log^{q^2}(n))$ samples.

		The first step is of the algorithm is to generate the $p$-stable vectors $Z^{i,j} \in \R^{n_i}$ for $i \in [q]$ and $j =1,2,\dots,\Theta(\log(n))$. We can pre-compute and store $Z^{i,j} A_i$ for $i \in [q]$, which takes $\tilde{O}(\sum_{i=1}^q \nnz(A_i))$ time. We set $w^j \leftarrow \left( (\mathbb{I}_{n_1}) \otimes  (\bigotimes_{k = 2}^{q} Z^{k,j}) \rho\right) \in \R^{n_1}$ and define $w \in \R^{n_1}$ by $w_l = \median_{j \in [\tau]} \{		|w^j_l|\}$ for $l \in [n_1]$. Observe that $w^j_l$ is an estimate of $\sum_{\vec{u} : u_1 = l} |\rho_{\vec{u}}|^p$. By Proposition \ref{prop:stablekronecker}, it is a $(c\log(n))^q$ approximation with probability at least $3/4$ for some constant $c$. Taking the median of $\Theta(\log(n))$ repetitions, we have that \begin{align*}
		c^{-q} \cdot \sum_{\vec{u} : u_1 = l} |\rho_{\vec{u}}|^p \leq |w_l|^p \leq (c\log(n))^{q} \cdot \sum_{\vec{u} : u_1 = l} |\rho_{\vec{u}}|^p 
		\end{align*}
		with probability $1-1/\poly(n)$, and we can then union bound over all such estimates every conducted over the course of the algorithm. We call the above estimate $|w_l|^p$ a $O( (c \log(n))^q)$-error estimate of $\sum_{\vec{u} : u_1 = l} |\rho_{\vec{u}}|^p$. Given this, we can correctly and independently sample the first coordinate of each of the $\Theta(r_3)$ samples. We now describe how to sample the $i$-th coordinate. So in general, suppose we have sampled $(a_1,...,a_{i-1})$ so far, and we need to now sample $a_i \in [n_i]$ conditioned on $(a_1,...,a_{i-1})$. We first consider
		\begin{equation*}
		W^{i,k} = \left( (\bigotimes_{k = 1}^{i-1} e_{a_k}^\top) \otimes (\mathbb{I}_{n_i}) \otimes  (\bigotimes_{k = i+1}^{q} Z^{k,j}) \rho\right) \in \R^{n_i}
		\end{equation*}
		
		Note that the $j$-th coordinate $W^{i,k}_j$ for $W^{i,k}$ is an estimate of $\sum_{\vec{u} \in \Delta(\vec{a}) : u_i = j} |\rho_{\vec{u}}|^p$. Again by By Proposition \ref{prop:stablekronecker}, with probability $1-1/\poly(n)$, we will have $|W^{i,k}_j|^p$ is a $O((c \log(n))^q)$-error estimate of $\sum_{\vec{u} \in \Delta(\vec{a}) : u_i = j} |\rho_{\vec{u}}|^p$ or at least one $k \in [\tau]$. Our goal will now be to find all $j \in [n_i]$ such that $\sum_{\vec{u} \in \Delta(\vec{a}) : u_i = j} |\rho_{\vec{u}}|^p \geq \Theta((c \log(n))^{q}/r_3^{8}) \sum_{\vec{u} \in \Delta(\vec{a})} |\rho_{\vec{u}}|^p$. We call such a $j$ a \textit{heavy hitter}.
		
		Let $Q_i \subset [n_i]$ be the set of heavy hitters. To find all the heavy hitters, we use the count-sketch for heavy hitters with the Dyadic trick of definition \ref{def:csHH}. We construct this count-sketch of def \ref{def:csHH} $S^i \in \R^{k' \times n_i}$ where $k' = O(\log^2(n)r_3^{16})$. We then compute $S^i W^{i,k}$, for $k=1,2,\dots,\tau$, and obtain the set of heavy hitters $h \in H_{i,k} \subset [n_i]$ which satisfy $|W^{i,k}_j|^p \geq \Theta(1/r_3^{8}) \|W^{i,k}\|_p^p$. By the above discussion, we know that for each $j \in Q_i$, we will have $|W^{i,k}_j|^p \geq \Theta(1/r_3^{16}) \|W^{i,k}\|_p^p$ for at least one $k \in [\tau]$ with high probability. Thus $H_i = \cup_{k=1}^\tau H_{i,k} \supseteq Q_i$.

		We now will decide to either sample a heavy hitter $ \xi \in H_i$, or a non-heavy hitter $\xi \in [n_i] \setminus H_i$. By Proposition \ref{prop:stablekronecker}, we can compute a $O((c\log(n))^{-q})$-error estimate  
		\begin{align*}
		\beta_i = \median_{j \in [\tau]} \left|\left( (\bigotimes_{k = 1}^{i-1} e_{a_k}^\top)  \otimes  (\bigotimes_{k = i}^{q} Z^{k,j}) \rho\right)\right|^p
		\end{align*}
		of $\sum_{\vec{u} \in \Delta(\vec{a})} |\rho_{\vec{u}}|^p$, meaning:
		
		\begin{equation*}
		O(c^{-q}) \sum_{\vec{u} \in \Delta(\vec{a})} |\rho_{\vec{u}}|^p \leq  \beta_i\leq O((c \log n)^{q}) \sum_{\vec{u} \in \Delta(\vec{a})} |\rho_{\vec{u}}|^p .
		\end{equation*}
		Again, by Proposition \ref{prop:stablekronecker}, we can compute a $O((c\log(n))^{-q})$-error estimate 
		\begin{align*}
		\gamma_i = \median_{j \in [\tau]}\left( (\bigotimes_{k = 1}^{i-1} e_{a_k}^\top) \otimes  Z^{i,j}_{[n_i] \setminus H_i} \otimes  (\bigotimes_{k = i+1}^{q} Z^{k,j}) \rho\right)
		\end{align*}
		of $\sum_{h \in [n_i] \setminus H_i}\sum_{\vec{u} \in \Delta(\vec{a}) : u_i = j} |\rho_{\vec{u}}|^p$. It follows that 
		
		\begin{align*}
		O(c^{-2q}) \frac{\sum_{h \in [n_i] \setminus H_i}\sum_{\vec{u} \in \Delta(\vec{a}) : u_i = j} |\rho_{\vec{u}}|^p}{\sum_{\vec{u} \in \Delta(\vec{a})} |\rho_{\vec{u}}|^p}		
		\leq 	\frac{\gamma_i}{\beta_i} \leq  O((c \log n)^{2q}) \frac{\sum_{h \in [n_i] \setminus H_i}\sum_{\vec{u} \in \Delta(\vec{a}) : u_i = j} |\rho_{\vec{u}}|^p}{\sum_{\vec{u} \in \Delta(\vec{a})} |\rho_{\vec{u}}|^p}	
		\end{align*}
		In other words, $\gamma_i/\beta_i$ is a $O((c\log(n))^{2q})$-error approximation of the true probability that we should sample a non-heavy item. Thus with probability $1-\gamma_i/\beta_i$, we choose to sample a heavy item.	
		
		To sample a heavy item, for each $\xi \in H_i$, by Proposition \ref{prop:stablekronecker}, we can compute an $O((c\log(n))^{-q})$-error estimate 
		\begin{align*}
		\median_{j \in \tau}\left|\left( (\bigotimes_{k = 1}^{i-1} e_{a_k}^\top) \otimes (e_\xi^\top) \otimes  (\bigotimes_{k = i+1}^{q} Z^{k,j}) \rho\right)\right|^p
		\end{align*}
		of $\sum_{\vec{u} \in \Delta(\vec{a}) : u_i = \xi} |\rho_{\vec{u}}|^p$, meaning 
		
		\begin{align*}
		\O(c^{-q}) \sum_{\vec{u} \in \Delta(\vec{a}) : u_i = \xi} |\rho_{\vec{u}}|^p  &\leq \median_{j \in \tau}\left|\left( (\bigotimes_{k = 1}^{i-1} e_{a_k}^\top) \otimes (e_\xi^\top) \otimes  (\bigotimes_{k = i+1}^{q} Z^{k,j}) \rho\right)\right|^p \\
		&\leq  O((c \log n)^{q}) \sum_{\vec{u} \in \Delta(\vec{a}) : u_i = \xi} |\rho_{\vec{u}}|^p 
		\end{align*}
		Thus we can choose to sample a heavy item $\xi \in H_i$ from the distribution given by 
		\begin{align*}
		\bpr{\text{sample } a_i \leftarrow \xi } =  \frac{\text{median}_{j \in \tau}\left|\left( (\bigotimes_{k = 1}^{i-1} e_{a_k}^\top) \otimes (e_\xi^\top) \otimes  (\bigotimes_{k = i+1}^{q} Z^{k,j}) \rho\right)\right|^p }{ \sum_{\xi' \in H_i} \text{median}_{j \in \tau}\left|\left( (\bigotimes_{k = 1}^{i-1} e_{a_k}^\top) \otimes (e_{\xi'}^\top) \otimes  (\bigotimes_{k = i+1}^{q} Z^{k,j}) \rho\right)\right|^p} 
		\end{align*}
		Which gives a $O((c\log(n))^{2q})$-error approximation to the correct sampling probability for a heavy item.

		In the second case, with probability $\gamma_i/\beta_i$, we choose to not sample a heavy item. In this case, we must now sample a item from $[n_i] \setminus H_i$. To do this, we partition $[n_i]$ randomly into $\Omega_1,\dots,\Omega_\eta$ for $\eta = 1/r_3^2$. Now there are two cases. First suppose that we have 
		
		\begin{align*}
		\frac{\sum_{j \in [n_i] \setminus H_i} \sum_{\vec{u} \in \Delta(\vec{a}) : u_i = j} |\rho_{\vec{u}}|^p }{\sum_{\vec{u} \in \Delta(\vec{a})} |\rho_{\vec{u}}|^p} \leq  \Theta(1/r_3^3)
		\end{align*}
		Now recall that $\gamma_i/\beta_i$ was a $O((c \log(n))^{2q})$-error estimate of the ratio on the left hand side of the above equation, and $\gamma_i/\beta_i$ was the probability with which we choose to sample a non-heavy hitter.
		Since we only repeat the sampling process $r_3$ times, the probability that we ever sample a non-heavy item in this case is at most $\Theta(q (c\log(n))^{2q}/r_3^2) < \Theta(q/r_3)$, taken over all possible repetitions of this sampling in the algorithm. Thus we can safely ignore this case, and condition on the fact that we never sample a non-heavy item in this case. 
		
		Otherwise,  
		\begin{align*}
		\sum_{j \in [n_i] \setminus H_i} \sum_{\vec{u} \in \Delta(\vec{a}) : u_i = j} |\rho_{\vec{u}}|^p >\Theta(1/r_3^3) \sum_{\vec{u} \in \Delta(\vec{a}) } |\rho_{\vec{u}}|^p,
		\end{align*}
		and it follows that 
		\begin{align*}
		\sum_{\vec{u} \in \Delta(\vec{a}) : u_i = j'} |\rho_{\vec{u}}|^p \leq r_3^{-5} \cdot \sum_{j \in [n_i] \setminus H_i} \sum_{\vec{u} \in \Delta(\vec{a}) : u_i = j} |\rho_{\vec{u}}|^p
		\end{align*}
		for all $j' \in [n_i] \setminus H_i$, since we removed all $\Theta(1/r_3^8)$ heavy hitters from $[n_i]$ originally. Thus by Chernoff bounds, with high probability we have that 
		\begin{align*}
		\sum_{j \in \Omega_i \setminus H_i} \sum_{\vec{u} \in \Delta(\vec{a}) : u_i = j} |\rho_{\vec{u}}|^p = \Theta \left( \eta^{-1} \cdot \sum_{j \in [n_i] \setminus H_i} \sum_{\vec{u} \in \Delta(\vec{a}) : u_i = j} |\rho_{\vec{u}}|^p \right),
		\end{align*}
		which we can union bound over all repetitions.

		Given this, by choosing $t \sim [\eta]$ uniformly at random, and then choosing $j \in \Omega_t \setminus H_i$ with probability proportional to its mass in $\Omega_t \setminus H_i$, we get a $\Theta(1)$ approximation of the true sampling probability. Since we do not know its exact mass, we instead sample from the distribution 
		\begin{align*}
		\left\{ \frac{\theta_j}{\sum_{j' \in \Omega_t \setminus H_i} \theta_{j'} } \right\}_{ j \in \Omega_t \setminus H_i },
		\end{align*}
		where  
		\begin{equation*}
		\theta_j = \median_{l \in [\tau]} \left(\left| (\bigotimes_{k = 1}^{i-1} e_{a_k}^\top) \otimes (e_j^\top) \otimes  (\bigotimes_{k = i+1}^{q} Z^{k,l}) \rho \right|^p\right)
		\end{equation*}
		
		Again by Proposition \ref{prop:stablekronecker}, this gives a $O((c \log(n))^{2q})$-error approximation to the correct sampling probability. Note that at each step of sampling a coorindate of $\vec{a}$ we obtained at most $O((c \log(n))^{2q})$-error in the sampling probability. Thus, by oversampling by a $O((c \log(n))^{2q^2})$ factor, we can obtain the desired sampling probabilities. This completes the proof of correctness. Note that to improve the failure probability to $1-\delta$, we can simply scale $r_3$ by a factor of $1/\delta$.

		\paragraph{Proof of Runtime.}
		We now analyze the runtime. At every step $i=1,2,\dots,q$ of the sampling, we compute $v_i^j \leftarrow S^i \left( (\bigotimes_{k = 1}^{i-1}e_{a_k}^\top ) \otimes (\mathbb{I}_{n_i}) \otimes  (\bigotimes_{k = i+1}^{q} Z^{k,j}) \rho\right) \in \R^{n_i}$ for $j =1,2,\dots\Theta(\log(n))$. This is equal to
		\[ S^i \left((\bigotimes_{k = 1}^{i-1} (A_k)_{a_k,*}) \otimes (A_i) \otimes  (\bigotimes_{k = i+1}^{q} Z^{k,j} A_k)x' - (\bigotimes_{k = 1}^{i-1} e_{a_k}^\top) \otimes (\mathbb{I}_{n_i}) \otimes  (\bigotimes_{k = i+1}^{q} Z^{k,j}) b\right) \]
		We first consider the term inside of the parenthesis (excluding $S^i$).
		Note that the term $(\bigotimes_{k = i+1}^{q} Z^{k,j} A_k)$ was already pre-computed, and is a vector of length at most $d$, this this requires a total of $\tilde{O}(\sum_{i=1}^q \nnz(A_i) + d)$ time. Note that these same values are used for every sample. Given this pre-computation, we can rearrage the first term to write $(\bigotimes_{k = 1}^{i-1} (A_k)_{a_k,*}) \otimes (A_i)  X'  (\bigotimes_{k = i+1}^{q} Z^{k,j} A_k)^\top$ where $X'$ is a matrix formed from $x'$ so that $x'$ is the vectorization of $X'$ (this is done via reshaping Lemma \ref{lem:reshaping}).  The term $y = X'  (\bigotimes_{k = i+1}^{q} Z^{k,j} A_k)^\top$ can now be computed in $O(d)$ time, and then we reshape again to write this as $(\bigotimes_{k = 1}^{i-1}(A_k)_{a_k,*}) Y A_i^\top$ where $Y$ again is a matrix formed from $y$. Observe that $\zeta = \text{vec}(\bigotimes_{k = 1}^{i-1}(A_k)_{a_k,*} Y) \in \R^{d_i}$ can be computed in time $O(qd)$, since each entry is a dot product of a column $Y_{*,j} \in \R^{d_1 \cdot d_2 \cdots d_{i-1}}$ of $Y$ with the $d_1 \cdot d_2 \cdots d_{i-1}$ dimensional vector $\bigotimes_{k = 1}^{i-1}(A_k)_{a_k,*}$, which can be formed in $O(d_1 \cdot d_2 \cdots d_{i-1}q)$ time, and there are a total of $d_i$ columns of $Y$.
		
		Given this, The first entire term $ S^i(\bigotimes_{k = 1}^{i-1} (A_k)_{a_k,*}) \otimes (A_i) \otimes  (\bigotimes_{k = i+1}^{q} Z^{k,j} A_k)x'$ can be rewritten as $S^iA_i \zeta$, where $\zeta = \zeta_{\vec{a}} \in \R^{d_i}$ can be computed in $O(dq)$ time for each sample $\vec{a}$. Thus if we recompute the value $S_i A_i \in \R^{k \times n}$, where $k = \wt{O}(r_3^{16})$, which can be done in time $\wt{O}(\nnz{A_i})$, then every time we are sampling the $i$-th coordinate of some $\vec{a}$, computing the value of $S^iA_i \zeta_{\vec{a}}$ can be done in time $O(k d_i^2) = r_3^{O(1)}.$
		
		We now consider the second term.		
		We perform
		 similar trick, reshaping $b \in \R^n$ into $B \in \R^{(n_1 \cdots n_{i}) \times(n_i \cdots n_q)}$ and writing this term as $ ((\bigotimes_{k = 1}^{i-1} e_{a_k}^\top) \otimes (\mathbb{I}_{n_i}) )B (\bigotimes_{k = i+1}^{q} Z^{k,j})^\top $ and computing $b' = B (\bigotimes_{k = i+1}^{q} Z^{k,j})^\top \in \R^{(n_1 \cdots n_{i}) }$ in $\nnz(B) = \nnz(b)$ time. Let $B' \in \R^{(n_1 \cdots n_{i-1})\times n_i}$ be such that $\text{vec}(B') = b'$, and we reshape again to obtain $(\bigotimes_{k = 1}^{i-1} e_{a_k}^\top) B' (\mathbb{I}_{n_i}) = (\bigotimes_{k = 1}^{i-1} e_{a_k}^\top) B'$ Now note that so far, the value $B'$ did not depend on the sample $\vec{a}$ at all. Thus for each $i=1,2,\dots,q$, $B'$ (which depends only on $i$) can be pre-computed in $\nnz(b)$ time. Given this, the value  $(\bigotimes_{k = 1}^{i-1} e_{a_k}^\top) B'$ is just a row $B_{(a_1,\dots,a_k),*}'$ of $B'$ (or a column of $(B')^\top$). We first claim that $\nnz(B') \leq \nnz(b) = \nnz(B)$. To see this, note that each entry of $B'$ is a dot product $ B_{j,*} (\bigotimes_{k = i+1}^{q} Z^{k,j})^\top$ for some row $B_{j,*}$ of $B$, and moreover there is a bijection between these dot products and entries of $B'$. Thus for every non-zero entry of $B'$, there must be a unique non-zero row (and thus non-zero entry) of $B$. This gives a bijection from the support of $B'$ to the support of $B$ (and thus $b$) which completes the claim. Since $S^i (B_{(a_1,\dots,a_k),*}')^\top$ can be computed in $\tilde{O}(\nnz(B_{(a_1,\dots,a_k),*}'))$ time, it follows that $S^i (B_{(a_1,\dots,a_k),*}')^\top$ can be computed for all rows $(B_{(a_1,\dots,a_k),*}')$ of $B$ in $\wt{O}(\nnz(b))$ time. Given this precomputation, we note that $(\mathbb{I}_{n_i}) \otimes  (\bigotimes_{k = i+1}^{q} Z^{k,j}) b$ is just $S^i(B_{(a_1,\dots,a_k),*}')^\top$ for some $(a_1,\dots,a_k)$, which has already been pre-computed, and thus requires no addition time per sample.  Thus, given a total of $\wt{O}(\sum_{i=1}^q \nnz(A_i)  + q\nnz(b) +r_3^{O(1)})$ pre-processing time, for each sample we can compute $v_i^j$ for all $i \in [q]$ and $j \in [\tau]$ in $\wt{O}(r_3^{O(1)})$ time, and thus  $\wt{O}(r_3^{O(1)})$ time over all $r_3$ samples.

		Given this, the procedure to compute the heavy hitters $H_{i,j}$ takes $\wt{O}(r_3^{16})$ time by Definition \ref{def:csHH} for each sample and $i \in [q], j \in [\tau]$. By a identical pre-computation and rearrangement argument as above, each $\beta_i^j$ (and thus $\beta_i$) can be computed in  $\wt{O}(r_3^{O(1)})$ time per sample after pre-computation. Now note that $\gamma_i$ is simply equal to 
		\begin{align*}
		\median_{j \in [\tau]} \left( \beta_i^j -  (\bigotimes_{k = 1}^{i-1} e_{a_k}^\top) \otimes (Z^{k,j}_{H_i}) \otimes  (\bigotimes_{k = i+1}^{q} Z^{k,j}) \rho \right).
		\end{align*}
		Since $Z^{k,j}_{H_i}$ is sparse, the above can similar be computed in $O(d|H_i|) = \wt{O}(r_3^{O(1)})$ time per sample after pre-computation. To see this, note that the $b$ term of $ (\bigotimes_{k = 1}^{i-1} e_{a_k}^\top) \otimes (Z^{k,j}_{H_i}) \otimes  (\bigotimes_{k = i+1}^{q} Z^{k,j}) \rho$ can be written as $(\bigotimes_{k = 1}^{i-1} e_{a_k}^\top)B''' (Z^{k,j}_{H_i})^\top$, where $B''' \in \R^{n_1 \cdots n_{i-1} \times n_i}$ is a matrix that has already been pre-computed and does not depend on the given sample. Then this quantity is just the dot product of a row of $B'''$ with  $(Z^{k,j}_{H_i})^\top$, but since $(Z^{k,j}_{H_i})$ is $|H_i|$-sparse, so the claim for the $b$ term follows. For the $(A_1 \otimes \cdots \otimes A_q)$ term, just as we demonstrated in the discussion of computing $v_i^j$, note that this can be written as $(\bigotimes_{k = 1}^{i-1}(A_k)_{a_k,*}) Y ((A_i)_{H_i,*})^\top$ for some matrix $Y \in \R^{d_1 \cdots d_i \times d_{i-1}}$ that has already been precomputed.  Since $(A_i)_{H_i,*}$ only has $O(|H_i|)$ non-zero rows, this whole product can be computed in time $O(d|H_i|)$ as needed.
		
		Similarly, we can compute the sampling probabilities 
		\begin{equation*}
		\bpr{\text{sample } a_i \leftarrow j } =  \frac{\text{median}_{j \in \tau}\left|\left( (\bigotimes_{k = 1}^{i-1} e_{a_k}^\top) \otimes (e_\xi^\top) \otimes  (\bigotimes_{k = i+1}^{q} Z^{k,j}) \rho\right)\right|^p }{ \sum_{\xi' \in H_i} \text{median}_{j \in \tau}\left|\left( (\bigotimes_{k = 1}^{i-1} e_{a_k}^\top) \otimes (e_{\xi'}^\top) \otimes  (\bigotimes_{k = i+1}^{q} Z^{k,j}) \rho\right)\right|^p} 
		\end{equation*}
		for each every item $\zeta \in H_i$ in $\wt{O}(r_3^{O(1)})$ time after pre-computation, and note $|H_i| =\wt{O}(r_3^{O(1)})$ by definition \ref{def:csHH}. Thus the total time to sample a heavy hitter in a given coordinate $i \in [q]$  for each sample $\wt{O}(r_3^{O(1)})$ per sample, for an overall time of $\wt{O}(qr_3^{O(1)})$ over all samples and $i \in [q]$.
		
		Finally, we consider the runtime for sampling a non-heavy item. Note that $|\Omega_t| = O(n_i/\eta)$ with high probability for all $t \in [\eta]$ by chernoff bounds. Computing each 
		\begin{align*}
		\theta_j = \median_{l \in [\tau]} \left(\left| (\bigotimes_{k = 1}^{i-1} e_{a_k}^\top) \otimes (e_j^\top) \otimes  (\bigotimes_{k = i+1}^{q} Z^{k,l}) \rho \right|^p\right)
		\end{align*}  takes $O(qd)$ time after pre-computation, and so we spend a total of $O(qd n_i/\eta)$ time sampling an item from $\Omega_t \setminus H_i$. Since we only ever sample a total of $r_3$ samples, and $\eta = \Theta(r_3^2)$, the total time for sampling non-heavy hitters over the course of the algorithm in coordinate $i$ is $o(n_i) = o(\nnz(A_i))$ as needed, which completes the proof of the runtime. 
		
		\paragraph{Computing the Sampling Probabilities $\alpha_i$}
		The above arguments demonstrate how to sample efficiently from the desired distribution. We now must describe how the sampling probabilities $\alpha_i$ can be computed. First note, for each sample that is sampled in the above way, at every step we compute exactly the probability with which we decide to sample a coordinate to that sample. Thus we know exactly the probability that we choose a sample, and moreover we can compute each $q_i$ in $O(d)$ time as in Lemma \ref{lem:O1approx}. Thus we can compute the maximum of $q_i$ and this probability exactly. For each item sampled as a result of the leverage score sampling probabilities $q_i$ as in Lemma \ref{lem:O1approx}, we can also compute the probability that this item was sampled in the above procedure, by using the same sketching vectors $Z^{i,k}$ and count-sketches $S^i$. This completes the proof of the Lemma.

	\end{proof}

\section{All-Pairs Regression}\label{sec:allpairs}

Given a matrix $A \in \R^{n \times d}$ and $b \in \R^n$, let $\bar{A} \in \R^{n^2 \times d}$ be the matrix such that $\bar{A}_{i + (j-1)n,*} = A_{i,*} - A_{j,*}$, and let $\bar{b} \in \R^{n^2}$ be defined by $\bar{b}_{i + (j-1)n} = b_{i} - b_{j}$. Thus, $\bar{A}$ consists of all pairwise differences of rows of $A$, and $\bar{b}$ consists of all pairwise differences of rows of $b$,. The $\ell_p$ all pairs regression problem on the inputs $A,b$ is to solve $\min_{x \in \R^{d}} \|\bar{A}x - \bar{b}\|_p$.

First note that this problem has a close connection to Kronecker product regression. Namely, the matrix $\bar{A}$ can be written $\bar{A} = A \otimes \1^n -\1^n \otimes A$, where $\1^n \in \R^n$ is the all $1$'s vector. Similarly, $\bar{b} = b \otimes \1^n - \1^n \otimes b$. For simplicity, we now drop the superscript and write $\1 = \1^n$.  

Our algorithm is given formally in Figure \ref{alg:apreg}. We generate sparse $p$-stable sketches $S_1,S_2 \in \R^{k \times n}$, where $k = (d/(\eps \delta))^{O(1)}$. We compute $M = (S_1 \otimes S_2)(F \otimes \1 - \1 \otimes F) = S_1 F \otimes S_2 \1 -  S_1\1 \otimes S_2 F$, where $F = [A,b]$. We then take the $QR$ decomposition $M = QR$. Finally, we sample rows of $(F \otimes \1 - \1 \otimes F)R^{-1}$ with probability proportional to their $\ell_p$ norms. This is done by an involved sampling procedure described in Lemma \ref{lem:samplefast}, which is similar to the sampling procedure used in the proof of Theorem \ref{thm:l1_regression}. Finally, we solve the regression problem $\min_x \|\Pi(\bar{A}x - \bar{b})\|_p$, where $\Pi$ is the diagonal row-sampling matrix constructed by the sampling procedure. 
We summarize the guarantee of our algorithm in the following theorem.

\begin{theorem}\label{thm:allpairsmain}
	Given $A \in \R^{n \times d}$ and $b \in \R^n$, for $p \in [1,2]$, let $\bar{A} = A \otimes \1 -\1 \otimes A \in \R^{n^2 \times d}$ and $\bar{b} = b \otimes \1 - \1 \otimes b \in \R^{n^2}$. Then there is an algorithm for that outputs $\hat{x} \in \R^d$ such that with probability $1-\delta$ we have	$\|\bar{A}\hat{x} - \bar{b}\|_p \leq (1+\eps) \min_{x \in \R^d}\|\bar{A}x -\bar{b}\|_p$.
	The running time is $\wt{O}(\nnz(A) + (d/(\eps\delta))^{O(1)})$.
\end{theorem}

\begin{algorithm*}[!t]\caption{Our All-Pairs Regression Algorithm}\label{alg:apreg}
	\begin{algorithmic}[1]
		\Procedure{All-Pairs Regression}{$A,b$} 
		\State $F = [A,b] \in \R^{n \times d+1}$. $r\leftarrow \poly(d/\eps)$
		\State Generate $S_1,S_2 \in \R^{k \times n}$ sparse $p$-stable transforms for $k = \poly(d/(\eps \delta))$. 
		\State Sketch $(S_1 \otimes S_2)(F \otimes \1 - \1 \otimes F)$.
		\State Compute $QR$ decomposition: $(S_1 \otimes S_2)(F \otimes \1 - \1 \otimes F)= QR$.
		\State Let $M = (F \otimes \1 - \1 \otimes F)R^{-1}$, and $\sigma_i = \|M_{i,*}\|_p^p/\|M\|_p^p$. 
		\State Obtain row sampling diagonal matrix $\Pi \in \R^{n \times n}$ such that $\Pi_{i,i}=1/\wt{q_i}^{1/p}$ independently with probability $q_i \geq \min\{1,r\sigma_i\}$, where $\wt{q_i} = (1 \pm \eps^2)q_i$.  \Comment{Lemma ~\ref{lem:samplefast}}
		\State \Return $\hat{x}$ , where $\hat{x} = \arg \min_{x \in \R^d} \|\Pi(\bar{A} x - \bar{b})\|_p$.
		\EndProcedure
	\end{algorithmic}
\end{algorithm*}
The theorem crucially utilizes our fast $\ell_p$ sampling routine, which is described in Figure \ref{alg:sample} in the supplementary. A full discussion and proof of the lemma can be found in the supplementary material \ref{sec:samplefastallpairs}.

\begin{lemma}[Fast $\ell_p$ sampling]\label{lem:samplefast}
	Given $R \in \R^{d+1 \times d+1}$ and $F = [A,b] \in \R^{n \times d+1}$, there is an algorithm that, with probability $1-n^{-c}$ for any constant $c$, produces a diagonal matrix $\Pi \in \R^{n^2 \times n^2}$ such that $\Pi_{i,i} = 1/
\wt{q_i}^{1/p}$ with probability $q_i \geq \min\{1,r \|M_{i,*}\|_p^p/\|M\|_p^p\}$ and  $\Pi_{i,i} = 0$ otherwise, where $r = \poly(d/\eps)$ and $M = (F \otimes \1 - \1 \otimes F)R^{-1}$, and $\tilde{q_i} = (1 \pm \eps^2)q_i$ for all $i \in [n^2]$. The total time required is $\wt{O}(\nnz{A} + \poly(d/\eps))$.
\end{lemma}

\subsection{Analysis of All-Pairs Regression Algorithm}

In this section, we prove the correctness of our all-pairs regression algorithm \ref{alg:apreg}. Our main theorem, Theorem \ref{thm:allpairsmain}, relies crucially on the sample routine developed in Section \ref{sec:samplefastallpairs}. We first prove the theorem which utilizes this routine, and defer the description and proof of the routine to Section \ref{sec:samplefastallpairs}. 

Recall first the high level description of our algorithm (given formally in Figure \ref{alg:apreg}).  We pick $S_1,S_2 \in \R^{k \times n}$and $S$ are sparse $p$-stable sketches. We then compute $M = (S_1 \otimes S_2)(F \otimes \1 - \1 \otimes F) = S_1 F \otimes S_2 \1 -  S_1\1 \otimes S_2 F$, where $F = [A,b]$. We then take the $QR$ decomposition $M = QR$. Finally, we sample rows of $(F \otimes \1 - \1 \otimes F)R^{-1}$ with probability proportional to their $\ell_p$ norms. This is done by the sampling procedure described in Section \ref{sec:samplefastallpairs}. Finally, we solve the regression problem $\min_x \|\Pi(\bar{A}x - \bar{b})\|_p$, where $\Pi$ is the diagonal row-sampling matrix constructed by the sampling procedure. 

We begin by demonstrating that  $S_1 \otimes S_2$ is a $\poly(d)$ distortion embedding for the column span of $[\bar{A},\bar{b}]$.

\begin{lemma}\label{lem:lowdistortionap}
	Let $S_1,S_2 \in \R^{k \times n}$ be sparse $p$-stable transforms, where $k = \poly(d/(\eps\delta))$. Then for all $x \in \R^{d+1}$, with probability $1-\delta$ we have
	\begin{align*}
	1/O(d^4 \log^4d) \cdot \| [\bar{A},\bar{b}]x\|_p \leq \|(S_1 \otimes S_2)[\bar{A},\bar{b}]x\|_p \leq O(d^2 \log^2 d) \cdot \|[\bar{A},\bar{b}]x\|_p .
	\end{align*}
\end{lemma}
\begin{proof}
	Let $F = [A,b]$. Then a basis for the columns of $[\bar{A},\bar{b}]$ is given by $F \otimes \1 - \1 \otimes F$. We first condition on both $S_1,S_2$ being a low-distortion embedding for the $d+2$ dimensional column-span of $[F,\1]$. Note that this holds with large constant probability by \ref{lem:lowdistembed}.
	
	So for any $x \in \R^{d+1}$, we first show the upper bound
	\begin{align*}
	\|(S_\1 \otimes S_2)(F \otimes \1 - \1 \otimes F)x\|_p 
	& = \|(S_1 F \otimes S_2 \1)x - (S_1 \1 \otimes S_2 F)\|_p \\
	&= \|S_1 F x \1^\top S_2^\top - S_1 \1 x^\top F^\top  S_2^\top \|_p \\
	&= \|S_1( F x \1^\top  -  \1 x^\top F^\top)  S_2^\top \|_p \\
	& \leq  O(d \log d) \cdot \|(F x \1^\top  -  \1 x^\top F^\top)  S_2^\top \|_p \\
	&\leq   O(d^2 \log^2 d) \cdot \|F x \1^\top  -  \1 x^\top F^\top \|_p \\
	&=  O(d^2 \log^2 d) \cdot \|(F \otimes \1 - \1 \otimes F)x\|_p , 
	\end{align*}
	where the first equality follows by properties of the Kronecker product \cite{l00}, the second by reshaping Lemma \ref{lem:reshaping}. The first inequality follows from the fact that each column of $( F x \1^\top  -  \1 x^\top F^\top)  S_2^\top$ is a vector in the column span of $[F,\1]$, and then using that $S_1$ is a low distortion embedding. The second inequality follows from the fact that each row of $(F x \1^\top  -  \1 x^\top F^\top)$ is a vector in the \textit{column span} of $[F,\1]$, and similarly using that $S_2$ is a low distortion embedding. The final inequality follows from reshaping. Using a similar sequence of inequalities, we get the matching lower bound as desired.
\end{proof}

We now prove our main theorem.

\paragraph{Theorem \ref{thm:allpairsmain}}{\it
	Given $A \in \R^{n \times d}$ and $b \in \R^n$, for $p \in [1,2]$ there is an algorithm for the All-Pairs Regression problem that outputs $\hat{x} \in \R^d$ such that with probability $1-\delta$ we have
	\[ 	\|\bar{A}\hat{x} - \bar{b}\|_p \leq (1+\eps) \min_{x \in \R^d}\|\bar{A}x -\bar{b}\|_p	\]
	Where  $\bar{A} = A \otimes \1 -\1 \otimes A \in \R^{n^2 \times d}$ and $\bar{b} = b \otimes \1 - \1 \otimes b \in \R^{n^2}$.
	For $p <2$, the running time is $\wt{O}(nd + (d/(\eps\delta))^{O(1)})$, and for $p=2$ the running time is $O(\nnz(A) + (d/(\eps\delta))^{O(1)})$.}
\begin{proof}
	We first consider the case of $p=2$. Here, we can use the fact that the \textsc{TensorSketch} random matirx $S \in \R^{k \times n}$ is a subspace embedding for the column span of $[\bar{A},\bar{b}]$ when $k = \Theta(d/\eps^2)$ \cite{dssw18}, meaning that $\|S[\bar{A},\bar{b}]\|_2 = (1 \pm \eps)\|[\bar{A},\bar{b}]x\|_2$ for all $x \in \R^{d+1}$ with probability $9/10$. Moreover, $S\bar{A}$ and $S\bar{b}$ can be computed in $O(\nnz(A) + \nnz(b)) = O(\nnz(A))$ by \cite{dssw18} since they are the difference of Kronecker products. As a result, we can simply solve the regression problem $\hat{x} = \arg\min_x \|S\bar{A}x -S \bar{b}\|_2$ in $\poly(kd)$ time to obtain the desired $\hat{x}$.
	
	For $p<2$, we use the algorithm in Figure \ref{alg:apreg}, where the crucial leverage score sampling procedure to obtain $\Pi$ in step 7 of Figure \ref{alg:apreg} is described in Lemma \ref{lem:samplefast}. Our high level approach follows the general $\ell_p$ sub-space embedding approach of \cite{ddhkm09}. Namely, we first compute a low-distortion embedding $(S_1 \otimes S_2)(F \otimes \1 - \1 \otimes F)$.  By Lemma \ref{lem:lowdistortionap}, using sparse-p stable transformations $S_1,S_2$, we obtain the desired $\poly(d)$ distortion embedding into $\R^{k^2}$, where $k = \poly(d/\eps)$. Note that computing $(S_1 \otimes S_2)(F \otimes \1 - \1 \otimes F)$ can be done in $O(\nnz(A) + \nnz(b) + n)$ time using the fact that $(S_1 \otimes S_2)(F \otimes \1) = S_1 F \otimes S_2 \1$. As shown in \cite{ddhkm09}, it follows that $M = (F \otimes \1 - \1 \otimes F)R^{-1}$ is an $\ell_p$  well-conditioned basis for the column span of $(F \otimes \1 - \1 \otimes F)$ (see definition \ref{def:wellconditioned}). Then by Theorem 5 of \cite{ddhkm09}, if we let $\hat{\Pi}$ be the diagonal row sampling matrix such that $\hat{\Pi}_{i,i} = 1/q_i^{1/p}$ for each $i$ with probability $q_i \geq \min\{1, r \|M_{i,*}\|_p^p/ \|M\|_p^p\}$ (and $\hat{\Pi}_{i,i} = 0$ otherwise) for $r = \poly(d \log(1/\delta)/\eps)$, then with probability $1-\delta$ we have 
	\begin{align*}
	\|\hat{\Pi} (F \otimes \1 - \1 \otimes F) x\|_p = (1 \pm \eps )\| (F \otimes \1 - \1 \otimes F) x\|_p
	\end{align*}
	for all $x \in \R^{d+1}$. First assume that we had such a matrix.

	Since $(\bar{A}x - \bar{b})$ is in the column span of $(F \otimes \1 - \1 \otimes F)$ for any $x \in \R^{d+1}$, it follows that $\|\hat{\Pi}(\bar{A}x - \bar{b})\|_p  = (1 \pm \eps)\|(\bar{A}x - \bar{b})\|_p$ for all $x \in \R^{d}$, which completes the proof of correctness. By Lemma \ref{lem:samplefast}, we can obtain a row sampling matrix $\Pi$ in time $\wt{O}(nd + \poly(d/\eps))$, except that the entries of $\Pi$ are instead equal to either $0$ or $1/\wt{q}_i^{1/p}$ where $\wt{q_i} = (1 \pm \eps^2)q_i$. Now let $\hat{\Pi}$ be the idealized row sampling matrices from above, with entries either $0$ or $1/q_i^{1/p}$ as needed for Theorem 5 of \cite{ddhkm09}. Note that for any matrix $Z$ each row of $\hat{\Pi}Z x$ is equal to $\Pi Z x$ times some constant $1-\eps^2<c < 1+ \eps^2$. It follows that 
	$\|\Pi(\bar{A}x - \bar{b})\|_p = (1 \pm \eps^2) \|\hat{\Pi}(\bar{A}x - \bar{b})\|_p$ for all $x \in \R^d$, and thus the objective function is changed by at most a $(1 \pm \eps^2)$ term, which is simply handled by a constant factor rescaling of $\eps$.
	
	Finally, we can solve the sketched regression problem $\|\Pi(\bar{A} x - \bar{b})\|_p$ which has $\poly(d/\eps)$ constraints and $d$ variables in time $\poly(d/\eps)$ using linear programming for $p=1$ (see \cite{ls14,ls15,cls19,lsz19,b20} for the state of the art linear program sovler), or more generally interior point methods for convex programming for $p> 1$ (see \cite{bcll18,akps19,lsz19} for the recent development of $\ell_p$ solver. Finally, the failure probability bound holds by union bounding over all the aforementioned results, and noting that the lowest probability event was the even that $S_1 \otimes S_2$ was a low distortion embedding via Lemma \ref{lem:lowdistortionap}.  This completes the proof of the theorem. 
	
\end{proof}

\begin{algorithm*}[!t]\caption{Algorithm to $\ell_p$ sample $\Theta(r)$ rows of $M=(F \otimes \1 - \1 \otimes F )R^{-1}$}\label{alg:sample}
	\begin{algorithmic}[1]
		\Procedure{$\ell_p$ sample}{$F= [A,b] \in \R^{n \times d},R^{-1} \in \R^{d+1 \times d+1}, r$} 
		\State Generate a matrix $G \in \R^{d+1 \times \xi}$ of i.i.d. $\mathcal{N}(0,1/\sqrt{\xi})$ Gaussian random variables, with $\xi = \Theta(\log(n))$
		\State $Y \leftarrow R^{-1} G \in \R^{d+1 \times \xi}$
		\State $C \leftarrow (F \otimes \1 - \1 \otimes F )Y$
		\State Reshape i-th column $C_{*,i}$ into $(F{Y_{*,i}} \1^\top - \1 (Y_{*,i})^\top F^\top) \in \R^{n \times n}$
		\State Generate $Z \in \R^{t \times n}$ i.i.d. $p$-stable for $t = \Theta(\log(n))$ \Comment{Definition \ref{def:dense_sketch}}
		\State  For all $(i,l) \in [\xi] \times [n]$, set $$\sigma_{i,l} \leftarrow \median_{\tau \in [t]}\left( \frac{\left|(Z (F{Y_{*,i}} \1^\top - \1 (Y_{*,i})^\top F^\top)_{\tau,l}\right|^p}{(\text{median}(\mathcal{D}_p))^p}\right)$$ \Comment{Indyk Estimator \cite{i06}}
		\State Set $W^{(i,l)}  \leftarrow (FY_{*,i} \1^\top - \1 Y_{*,i}^\top F^\top )_{*,l} = FY_{*,i} - \1 (FY)_{l,i} \in \R^n $
		\For{$ j = 1,\dots,\Theta(r)$}
		\State Sample $(i,l)$ from distribution $\sigma_{i,l}/\left(\sum_{i',l'}\sigma_{i',l'}\right)$.
		\EndFor
		\State $T \leftarrow$ multi-set of samples $(i,l)$ 
		\State Generate $S_0\in \R^{k \times n}$ $S \in \R^{k' \times n}$ count-sketches for heavy hitters with $k =r^{O(1)}, k' = k^{O(1)}$. \Comment{Definition \ref{def:csHH}}
		\State Generate $u_1,\dots,u_n$ i.i.d. exponential variables.
		\State $D \leftarrow  \texttt{Diag}(1/u_1^{1/p},\dots,1/u_n^{1/p}) \in \R^{n \times n}$.
		
		\For{each sample $(i,l) \in T$}
		\State Compute $S_0 W^{(i,l)}$ and obtain set of heavy hitters $Q_0^{(i,l)} \subset[n]$
		\State Compute $W^{(i,l)}_j$ exactly for all $j \in Q_0^{(i,l)}$, to obtain true heavy hitters $H^{(i,l)}$.
		\State Compute  $$\alpha_{i,l} \leftarrow \median_{\tau \in [t]}\left( \frac{\left|Z_{\tau,*} W^{(i,l)}- \sum_{\zeta \in H^{(i,l)}} Z_{\tau,\zeta} W_\zeta^{(i,l)} \right|^p }{(\text{median}(\mathcal{D}_p))^p}\right)$$ 
		\If{ With prob $1- \alpha_{(i,l)}/\sigma_{(i,l)}$, sample a heavy item $j^* \leftarrow j$}
		\State Sample a heavy item $j^* \leftarrow j$ from the distribution $|W^{(i,l)}_j|^p/\sum_{j \in H_{(i,l)} }|W^{(i,l)}_j|^p$.
		\State \Return The row $((l-1)n + j^*)$ \Comment{Note that $C_{(l-1)n + j^*, *}$ contains $W^{(i,l)}_{j^*}$} 
		\Else
		\State Randomly partition $[n]$ into $\Omega_1,\Omega_2,\dots,\Omega_\eta$ with $\eta = \Theta(r^4/\eps^4)$.
		\State Sample $t \sim [\eta]$ uniformly at random.
		\State Compute $S(DW^{(i,l)})_{\Omega_t \setminus H^{(i,l)}}$, and set $Q^{(i,l)} \subset \Omega_t \setminus H^{(i,l)}$ of heavy hitters.
		\State $j^* \leftarrow \arg \max_{j \in Q^{(i,l)}} (DW^{(i,l)})_j$
		\State \Return The row $((l-1)n + j^*)$ \Comment{Note that $C_{(l-1)n + j^*, *}$ contains $W^{(i,l)}_{j^*}$} 
		\EndIf
		\EndFor
		
		\EndProcedure
	\end{algorithmic}
\end{algorithm*}

\subsection{Proof of Fast Sampling Lemma \ref{lem:samplefast}}\label{sec:samplefastallpairs}
We now provide a full proof of the main technical lemma of Section \ref{sec:allpairs}. The sampling algorithm is given formally in Algorithm \ref{alg:sample}. The following proof of Lemma \ref{lem:samplefast} analyzes each step in the process, demonstrating both correctness and the desired runtime bounds. \\
{\bf Lemma \ref{lem:samplefast} }{\it   
	Given $R \in \R^{(d+1) \times (d+1)}$ and $F = [A,b] \in \R^{n \times (d+1)}$, there is an algorithm that, with probability $1-\delta$ for any $\delta > n^{-c}$ for any constant $c$, produces a diagonal matrix $\Pi \in \R^{n^2 \times n^2}$ such that $\Pi_{i,i} = 1/
	\wt{q_i}^{1/p}$ with probability $q_i \geq \min\{1,r \|M_{i,*}\|_p^p/\|M\|_p^p\}$ and  $\Pi_{i,i} = 0$ otherwise, where $r = \poly(d/\eps)$ and $M = (F \otimes \1 - \1 \otimes F)R^{-1}$, and $\tilde{q_i} = (1 \pm \eps^2)q_i$ for all $i \in [n^2]$. The total time required is $\wt{O}(\nnz{A} + \poly(d/\eps))$.}
\begin{proof}
	Our proof proceeds in several steps. We analyze the runtime concurrently with out analysis of correctness. 
	\paragraph{Reducing the number of Columns of $R^{-1}$. }
	We begin by generating a matrix $G \in \R^{(d+1) \times \xi}$ of i.i.d. $\mathcal{N}(0,1/\sqrt{\xi})$ Gaussian random variables. We then compute $Y \leftarrow R^{-1}G$ in $\wt{O}(d^2)$ time. We first claim that it suffices to instead $\ell_p$ sample rows of $C = (F \otimes \1 - \1 \otimes F)Y = MG$. Note that each entry $|C_{i,j}|^p$ is distributed as $g^p\|M_{i,*}\|_2^p$ where $G$ $\mathcal{N}(0,1/\sqrt{\xi})$ Gaussian, which holds by the $2$-stability of Gaussian random variables.  Note that $\ex{|g|^p} = \Theta(1/\xi)$, so $\ex{\|C_{i,*}\|_p^p} = \|M_{i,*}\|_2^p$, and by sub-exponential concentration (see Chapter 2 of \cite{wain19}), we have that $\|C_{i,*}\|_p^p = (1 \pm 1/10) \|M_{i,*}\|_2^p$ with probability $1-1/\poly(n)$, and we can union bound over this holding for all $i \in [n^2]$. By relationships between the $p$ norms, we have $\|M_{i,*}\|_p^p/d < \|M_{i,*}\|_2^p < \|M_{i,*}\|_p^p$, thus this changes the overall sampling probabilities by a factor between $\Theta(1/d^2)$ and $\Theta(d^2)$. Thus, we can safely oversample by this factor (absorbing it into the value of $r$) to compensate for this change in sampling probabilities.

	\paragraph{Sampling a row from $C$. }
	To sample a row from $C$, the approach will be to sample an entry $C_{i,j}$ of $C$ with probability proportional to  $\|C_{i,j}\|_p^p/\|C\|_p^p$. For every $(i,j)$ sampled, we sample the entire $i$-th row of $j$, so that the $j$-th row is indeed sampled with probability proportional to its norm. Thus, it suffices to sample entries of $C$ such that each $C_{i,j}$ is chosen with probability at least $\min\{1,r \|C_{i,j}\|_p^p/\|C\|_p^p\}$.
	First note that the $i$-th column of $C = (F \otimes \1 - \1 \otimes F)Y$ can be rearranged into a $n \times n$ matrix via Lemma \ref{lem:reshaping}, given by  $(F Y_{*,i} \1^\top - \1 Y^\top_{*,i} F^\top)$. To $\ell_p$ sample a coordinate from $C$, it suffices to first $\ell_p$ sample a column of one of the above matrices, and then $\ell_p$ sample an entry from that column. 
	
	To do this, we first compute $FY \in \R^{n \times \xi}$, which can be done in time $\wt{O}(\nnz{A})$ because $Y$ only has $\xi = \Theta(\log(n))$ columns. We then compute $Z(FY_{*,i}\1^\top - \1 Y_{*,i}^\top F^\top ) \in \R^{1 \times n}$ for all $i \in [d]$, where $Z \in \R^{1 \times n}$ is a fixed vector of i.i.d. $p$-stable random variables. Once $FY$ has been computed, for each $i \in [\xi]$ it takes $O(n)$ time to compute this $n$-dimensional vector, thus the total time required to compute all $\xi$ vectors is $\wt{O}(n)$. We repeat this process $t=O(\log(n))$ times with different $p$-stable vectors $Z^1,\dots,Z^\top$, and take the median of each coordinate of $Z^{j}(FY_{*,i} \1^\top - \1 Y_{*,i}^\top F^\top ) \in \R^{n}$, $j \in [t]$, divided by the median of the $p$-stable distribution (which can be approximated to $(1 \pm \eps)$ error in $\poly(1/\eps)$ time, see Appendix A.2 of \cite{kane2010exact} for details of this). This is done in Step $7$ of Algorithm \ref{alg:sample}. It is standard this this gives a $(1 \pm 1/10)$ approximation the the norm $\|(FY_{*,i} \1^\top - \1 Y_{*,i}^\top F^\top )_{*,l}\|_p$ for each $i \in [d], l \in [n]$ with probability $1-1/\poly(n)$ (See the Indyk median estimator \cite{i06}).  
	
	Now let $\sigma_{i,l}$ be our estimate of the norm $\|(FY_{*,i} \1^\top - \1 Y_{*,i}^\top F^\top )_{*,l}\|_p$, for all $i \in [\xi]$ and $l \in [n]$. We now  sample a columns $(i,l) \in [\xi] \times [n]$, where each $(i,l)$ is chosen with probability $\sigma_{i,l}/(\sum_{i',l'} \sigma_{i',l'})$. We repeat this process $\Theta(r)$ times, to obtain a \textit{multi-set} $T \subset [\xi] \times [n]$ of sampled columns $(i,l)$. We stress that $T$ is a multi-set, because the same column $(i,l)$ may have been chosen for multiple samples, and each time it is chosen we must independently sample one of the entries of that column. For any $(i,l) \in T$, we define $W^{(i,l)} =(FY_{*,i} \1^\top - \1 Y_{*,i}^\top F^\top )_{*,l} = (FY_{*,i} - \1 (FY)_{l,i} )$.
	
	\paragraph{$\ell_p$ Sampling an entry from $W^{(i,l)}$.}
	Now fix any $(i,l) \in T$. We show how to $\ell_p$ sample an entry from the vector $W^{(i,l)} \in \R^{n}$. In other words, for a given $j \in [n]$, we want to sample $W_j^{(i,l)} \in [n]$ with probability  at least $r|W_j^{(i,l)}|^p/\|W^{(i,l)}\|_p^p$. We do this in two steps. First, let $S_0 \in \R^{k \times n}$ be  the count-sketch for heavy hitters of definition \ref{def:csHH}, where $k = \poly(r)$. Note that we can compute $S_0F Y$ and $S_0 \1$ in time $\tilde{O}(n)$, since $FY \in \R^{n \times \xi}$. Once this is done, for each $(i,l) \in T$ we can compute $S_0 W^{(i,l)}$ in $O(k)$ time by computing $(S_0 \1 (FY)_{l,i})$ (note that $FY$ and $S_0\1$ are already computed), and subtracting it off from the $i$-th column of $S_0 F Y$, so the total time is $\tilde{O}(n + \poly(d/\eps))$ to compute $S_0W^{(i,l)}$ for all $(i,l) \in |T|$. Now we can obtain the set $Q^{(i,l)}_0 \subset [n]$ containing all the $\tilde{\Omega}(1/\sqrt{k})$ heavy hitters in $W^{(i,l)}$ with high probability. We can then explicitly compute the value of $W^{(i,l)}_j$ for all $j \in Q^{(i,l)}_0$, and exactly compute the set 
	\begin{align*}
	H^{(i,l)} = \left\{ j \in [n] ~\Big|~ | W^{(i,l)}_j|^p > \beta/r^{16} \|W^{(i,l)}\|_p^p \right\},
	\end{align*} 
	all in $\tilde{O}(k)$ time via definition \ref{def:csHH}, where $\beta > 0$ is a sufficiently small constant (here we use the fact that $|x|_p \geq |x|_2$ for $p \leq 2$). Note that we use the same sketch $S_0$ to compute all sets $Q^{(i,l)}_0$, and union bound the event that we get the heavy hitters over all $\poly(d/\eps)$ trails. 
	
	We are now ready to show how we sample an index from $W^{(i,l)}$. First, we estimate the total $\ell_p$ norm of the items in $[n_i] \setminus H^{(i,l)}$ (again with the Indyk median estimator), and call this $\alpha_{(i,l)}$ as in Algorithm \ref{alg:sample}, which can be computed in $O(|H^{(i,l)}|)$ additional time (by subtracting off the $|H^{(i,l)}|$ coordinates $ZW^{(i,l)}_\zeta$ for all heavy hitters $\zeta \in H^{(i,l)}$ from our estimate $\sigma_{(i,l)}$), and with probability $\alpha_{(i,l)}/\sigma_{(i,l)}$, we choose to sample one of the items of $H^{(i,l)}$, which we can then sample from the distribution $ |W^{(i,l)}_j|^p/(\sum_{j \in H^{(i,l)}} |W^{(i,l)}_j|^p)$. Since all the $\sigma_{(i,l)}, \alpha_{(i,l)}$'s were constant factor approximations, it follows that we sampled such an item with probability $\Omega(r|W^{(i,l)}_{j'}|^p/\|C\|_p^p)$ as needed.  Otherwise, we must sample an entry from $[n] \setminus H^{(i,l)}$. 
	To do this, we first randomly partition $[n]$ into $\eta = \Theta(r^4/\eps^4)$ subsets $\Omega_1,\Omega_2,\dots,\Omega_\eta$. 
	
	We now make the same argument made in the proof of Lemma \ref{lem:lpsampling}, considering two cases. In the first case, the $\ell_p$ mass of $[n] \setminus H^{(i,l)}$  drops by a $1/r^2$ factor after removing the heavy hitters. In this case, $\alpha_{(i,l)}/\sigma_{(i,l)}=O(1/r^2)$, thus we will never \textit{not} sample a heavy hitter with probability $1 - O(1/r)$, which we can safely ignore. Otherwise, the $\ell_p$ drops by less than a $1/r^2$ factor, and it follows that all remaining items must be at most a $\beta/r^{14}$ heavy hitter over the remaining coordinates $[n] \setminus H^{(i,l)}$ (since if they were any larger, they would be  $\beta/r^{16}$ heavy hitters in $[n]$, and would have been removed in $H^{(i,l)}$). Thus we can assume we are in the second case. 	
	So by Chernoff bounds, we have $\sum_{j \in \Omega_t} |W_j^{(i,l)}|_p = \Theta(\frac{1}{\eta} \sum_{j \in [n] \setminus H^{(i,l)}} |W_j^{(i,l)}|_p)$ with probability greater than $1- \exp(-\Omega(r))$. We can then union bound over this event occurring for all $t \in [\eta]$ and all $(i,l) \in T$. Given this, if we uniformly sample a $t \sim [\eta]$, and then $\ell_p$ sample a coordinate $j \in \Omega_t$, we will have sampled this coordinate with the correct probability up to a constant factor. We now sample such a $t$ uniformly from $\eta$.

	To do this, we generate a diagonal matrix $D \in \R^{n \times n}$, where $D_{i,i} = 1/u_{i}^{1/p}$, where $u_1,\dots,u_n$ are i.i.d. exponential random variables. For any set $\Gamma \subset [n]$, let $D_{\Gamma}$ be $D$ with all diagonal entries $(j,j)$ such that $j \notin \Gamma$ set equal to $0$. Now let $S \in \R^{k' \times n}$ be a second instance of count-sketch for heavy hitters of definition \ref{def:csHH}, where we set $k' = \poly(k)$ from above. It is known that returning $j^* = \arg \max_{j \in \Omega_t \setminus H^{(i,l)}} |(DW^{(i,l)})_j|$ is a perfect $\ell_p$ sample from $\Omega_t \setminus H^{(i,l)}$ \cite{jayaram2018perfect}. Namely, $\pr{j^* = j} =  |W^{(i,l)}_j|^p/\|W_{\Omega_t \setminus H^{(i,\ell)}}\|_p^p$ for any $j \in\Omega_t \setminus H^{(i,\ell)}$ . Thus it will suffice to find this $j^*$. To find $j^*$, we compute $S(DW)_{\Omega_t \setminus H^{(i,\ell)}}$. Note that since $FY$ has already been computed, to do this we need only compute $SD_{\Omega_t \setminus H^{(i,\ell)}}FY_{*,i}$ and $SD_{\Omega_t \setminus H^{(i,\ell)}}\1(FY)_{\ell,i}$, which takes total time $\tilde{O}(|\Omega_t \setminus H^{(i,\ell)}|)=\tilde{O}(n/\eta)$. We then obtain a set $Q^{(i,l)} \subset \Omega_t \setminus H^{(i,\ell)}$ which contains all $j$ with $|(DW^{(i,l)})_j| \geq \tilde{\Omega}(1/\sqrt{k'})\|(DW)_{\Omega_t \setminus H^{(i,\ell)}}\|_2$.
	
	As noted in \cite{jayaram2018perfect}, the value $\max_{j \in \Omega_t \setminus H^{(i,l)}} |(DW^{(i,l)})_j|$ is distributed identically to 
	\begin{align*}
	\|W_{\Omega_t \setminus H^{(i,\ell)}}\|_p/u^{1/p}
	\end{align*}
	where $u$ is again an exponential random variable. Since exponential random variables have tails that decay like $e^{-\Omega(x)}$, it follows that with probability $1-\exp(-\Omega(r))$ that we have 
	\begin{align*}
	\max_{j \in \Omega_t \setminus H^{(i,l)}} |(DW^{(i,l)})_j| = \Omega(\|W_{\Omega_t \setminus H^{(i,\ell)}}\|_p/r),
	\end{align*}
	and we can then union bound over the event that this occurs for all $(i,l) \in T$ and $\Omega_t$. Given this it follows that $(DW^{(i,l)})_{j^*}  = \Omega(\|W_{\Omega_t \setminus H^{(i,\ell)}}\|_p/r)$. Next, for any constant $c \geq 2$, by Proposition 1 of \cite{jayaram2018perfect}, we have $\|((DW)_{\Omega_t \setminus H^{(i,\ell)}})_{\texttt{tail}(c\log(n))}\|_2= \tilde{O}(\|W^{(i,l)}_{\Omega_t \setminus H^{(i,\ell)}}\|_p)$ with probability $1-n^{-c}$, where for a vector $x$, $x_{\tail(t)}$ is $x$ but with the top $t$ largest (in absolute value) entries set equal to $0$. Since there are at most $c\log(n)$ coordinates in $(DW)_{\Omega_t \setminus H^{(i,\ell)}}$ not counted in $((DW)_{\Omega_t \setminus H^{(i,\ell)}})_{\texttt{tail}(c\log(n))}$, and since $(DW)_{j^*}$ is the largest coordinate in all of $(DW)_{\Omega_t \setminus H^{(i,\ell)}}$, by putting together all of the above it follows that $(DW)_{j^*}$ is a $\tilde{\Omega}(1/r)$-heavy hitter in $(DW)_{\Omega_t \setminus H^{(i,\ell)}}$. Namely, that $|(DW)_{j^*}| \geq \tilde{\Omega}(\|(DW)_{\Omega_t \setminus H^{(i,\ell)}}\|_2/r)$. Thus, we conclude that $j^* \in Q^{(i,l)}$.

	Given that $j^* \in Q^{(i,l)}$, we can then compute the value $(DW^{(i,l)})_j = D_{j,j}(F Y_{j,i} - FY_{l,i})$ in $O(1)$ time to find the maximum coordinate $j^*$. Since $|Q^{(i,l)}| = O(k') = O(\poly(d/\eps))$, it follows that the total time required to do this is $\tilde{O}(n/\eta + \poly(d/\eps))$. Since we repeat this process for each $(i,l) \in T$, and $|T| = \Theta(r)$ whereas $\eta = \Theta(r^4)$, it follows that the total runtime for this step is $\tilde{O}(n/r^3 + \poly(d/\eps))$. By \cite{jayaram2018perfect}, the result is a perfect $\ell_p$ sample from $(DW)_{\Omega_t \setminus H^{(i,\ell)}}$, which is the desired result. 
	To complete the proof, we note that the only complication that remains is that we utilize the same scaling matrix $D$ to compute the sampled used in each of the columns $W^{(i,l)}$ for each $(i,l) \in T$. However, note that for $t \neq t'$, we have that $D_{\Omega_t}$ and $D_{\Omega_t}$ are independent random variables. Thus it suffices to condition on the fact that the $t \in [\eta]$ that is sampled for each of the $|T|$ repetitions of sampling a $\Omega_t$ are distinct. But this occurs with probability at least $1/r$, since $|T| = \Theta(r)$ and $\eta = \Theta(r^4)$. Conditioned on this, all $|T|$ samples are independent, and each sample is an entry $C_{i,j}$ of $C$ such that the probability that a given $(i,j)$ is chosen is $|C_{i,j}|^p/\|C\|_p^p$. Repeating this sampling $\Theta(r)$ times, we get that each $C_{i,j}$ is sampled with probability at least $\min\{1, r|C_{i,j}|^p/\|C\|_p^p\}$, which completes the proof of correctness. Note that the dominant runtime of the entire procedure was $\wt{O}(\nnz(A) + \poly(d/\eps))$ as stated, and the probability of success was $1- \exp(-r) + 1/\poly(n)$, which we can be amplified to any $1-\delta$ for $\delta > 1/n^c$ for some constant $c$ by increasing the value of $r$ by $\log(1/\delta)$ and the number of columns of the sketch $G$ to $\log(1/\delta)$, which does not effect the $\wt{O}(\nnz(A) + \poly(d/\eps))$ runtime.

	\paragraph{Computing approximations $\wt{q_i}$ for $q_i$.}
	It remains now how to compute the approximate sampling probabilities $\tilde{q}_i$ for $\Theta(r)$ rows of $C$ that were sampled. Note that to sample an entry, in $C$, we first sampled the $n \times 1$ submatrix $W^{(i,l)}$ of $C$ which contained it, where the probability that we sample this submatrix is known to us. Next, if the entry of $C$ was a heavy hitter in $W^{(i,l)}$, we exactly compute the probability that we sample this entry, and sample it with this probability. If the entry $j$ of $W^{(i,l)}$ is not a heavy hitter, we first sample an $\Omega_t$ uniformly with probability exactly $1/\eta$. The last step is sampling a coordinate from $W^{(i,l)}_{\Omega_t \setminus H^{(i,l)}}$ via exponential scaling. However, we do not know the exact probability of this sampling, since this will be equal to $|W^{(i,l)}_j|^p/ \|W^{(i,l)}_{\Omega_t \setminus H^{(i,l)}}\|_p^p$, and we do not know $\|W^{(i,l)}_{\Omega_t \setminus H^{(i,l)}}\|_p^p$ exactly. Instead, we compute it approximately to error $(1 \pm \eps^2)$ as follows. For each $(i,l) \in T$ and $\alpha = 1,2,\dots,\Theta(\log(n)/\eps^4)$, we compute $Z^{(\alpha)}W^{(i,l)}_{\Omega_t \setminus H^{(i,l)}}$, where $Z \in \R^{1 \times |\Omega_t \setminus H^{(i,l)}|}$ is a vector of $p$-stable random variables. Again, we use the Indyk median estimator \cite{i06}, taking the median of these $\Theta(\log(n)/\eps^4)$ repetitions,  to obtain 
	an estimate of $\|W^{(i,l)}_{\Omega_t \setminus H^{(i,l)}}\|_p^p$ with high probability to $(1 \pm \eps^2)$ relative error. Each repetition requires $O(|\Omega_t \setminus H^{(i,l)}|)$ additional time, and since $|\Omega_t \setminus H^{(i,l)}||T| = o(\eps^4n/r^3)$, it follows that the total computational time is at most an additive $o(n)$, thus computing the $\tilde{q}_i$'s to error $(1 \pm \eps^2)$ does not effect the overall runtime.

\end{proof}

\newpage
\section{Low Rank Approximation of Kronecker Product Matrices}\label{sec:lowrank}
We now consider low rank approximation of Kronecker product matrices. Given $q$ matrices $A_1,A_2,\dots$, $A_q$, where $A_i \in \R^{n_i \times d_i}$, the goal is to output a rank-$k$ matrix $B \in \R^{n \times d}$, where $n = \prod_{i=1}^q n_i$ and $d = \prod_{i=1}^q d_i$, such that $\| B - A \|_{F} \leq (1 + \eps) \OPT_k$,
 where $\OPT_k = \min_{\texttt{rank}-k~A'} \| A' - A \|_{F}$, and $A =\otimes_{i=1}^q A_i$. Our approach employs the Count-Sketch distribution of matrices \cite{cw13,w14}. A count-sketch matrix $S$ is generated as follows. Each column of $S$ contains exactly one non-zero entry. The non-zero entry is placed in a uniformly random row, and the value of the non-zero entry is either $1$ or $-1$ chosen uniformly at random. 

Our algorithm is as follows. We sample $q$ independent Count-Sketch matrices $S_1,\dots S_q$, with $S_i \in \R^{k_i \times n_i}$, where $k_1 = \dots = k_q =  \Theta(qk^2/\eps^2)$. We then compute $M = (\otimes_{i=1}^q S_i)A$, and let $U \in \R^{k \times d}$ be the top $k$ right singular vectors of $M$. Finally, we output $B = A U^\top  U$ in factored form (as $q+1$ separate matrices, $A_1,A_2,\dots,A_q,U$), as the desired rank-$k$ approximation to $A$. The following theorem demosntrates the correctness of this algorithm. 

\begin{theorem}\label{thm:LRAmain}
For any constant $q \geq 2$, there is an algorithm which runs in time $O(\sum_{i=1}^q \nnz(A_i) + d\poly(k/\eps))$ and outputs a rank $k$-matrix $B$ in factored form such that 
\begin{align*}
\| B - A \|_{F} \leq (1+\eps) \OPT_k
\end{align*} 
with probability $9/10$. \footnote{To amplify the probability, we can sketch $A$ and $A U^\top U$ with a sparse JL matrix (e.g., Lemma \ref{lem:subspace_amp} with $k_i = \Theta(q k^2/(\delta\eps^2))$ for each $i$)  in input sparsity time to estimate the cost of a given solution. We can then repeat $\log(1/\delta)$ times and take the minimum to get failure probability $1-\delta$.} 
\end{theorem}


\begin{proof}
	By Lemma \ref{lem:pcp}, we have $(1-\eps)\|A - AP\|_F^2 \leq \|M -M P\|_F^2 + c \leq (1+\eps)\|A - AP\|_F^2$ for all rank $k$ projection matrices $P$. In particular, we have
	\[  \min_{P}  (1+\eps)\|A - AP\|_F^2 + c = (1 + \eps)\OPT_k^2 \]
	where the minimum is taken over all rank $k$ projection matrices. The minimizer $P$ on the LHS is given by the projection onto the top $k$ singular space of $M$. Namely, $MP = M U^\top  U$ where $U$ is the top $k$ singular row vectors of $M$. Thus $ \|M - M U^\top U \|_F^2 + c \leq (1+\eps)\OPT_k^2$. Moreover, we have $\|A - A U^\top U \|_F^2 \leq (1 + 2\eps) ( \|M - M U^\top  U \|_F^2 + c) \leq (1+4\eps)\OPT_k^2$. Thus $\|A - A U^\top U \|_F \leq (1 + O(\eps))\OPT_k$ as needed. 
	
	For runtime, note that we first must compute $M= (\otimes_{i=1}^q S_i)(A_1 \otimes A_2) = S_1 A_1 \otimes \dots \otimes S_q A_q$. Now $S_i A_i$ can be computed in $O(\nnz(A_i))$ time for each $i$ \cite{cw13}. One all $S_iA_i$ are computed, their Kronecker product can be computed in time $O(q k_1 k_2 \cdots k_q d)= \poly(k d/\eps)$.  Given $M \in \R^{ k_1\cdots k_q \times d}$, the top $k$ singular vectors $U$ can be computed by computing the SVD of $M$, which is also done in time $\poly(k d/\eps)$. Once $U$ is obtained, the algorithm can terminate, which yields the desired runtime.
\end{proof}

To complete the proof of the main theorem, we will need to prove Lemma \ref{lem:pcp}. To do this, we begin by introducing two definitions.

\begin{definition}[Subspace embedding]
	A random matrix $S$ is called a $\eps$-subspace embedding for a rank $k$ subspace $\mathcal{V}$ we have simultaneously for all $x \in \mathcal{V}$ that 
	\begin{align*}
	\|Sx\|_2 = (1 \pm \eps)\|x\|_2.
	\end{align*} 
\end{definition}

\begin{definition}[Approximate matrix product]
	A random matrix $S$ satisfies the $\eps$-approximate matrix product property if, for any fixed matrices $A,B$, of the appropriate dimensions, we have 
	\begin{align*}
	\pr{\|A^\top S^\top S B - A^\top B\|_F \leq \eps \|A\|_F\|B\|_F} \geq 9/10.
	\end{align*} 
\end{definition}

We now show that $S$ is both a subspace embedding and satisfies approximate matrix product, where $S= \otimes_{i=1}^q S_i$ and $S_i \in \R^{k_i \times n_i}$ are count-sketch matrices.
\begin{lemma}\label{lem:subspace_amp}
	If $S =(\otimes_{i=1}^q S_i)$ with $S_i \in \R^{k_i \times n_i}$, $k_1 = k_2 = \dots = k_q =  \Theta(qk^2/\eps^2)$, then $S$ is an $\eps$-subspace embedding for any fixed $k$ dimensional subspace $\mathcal{V} \subset \R^n$ with probability $9/10$, and also satisfies the $(\eps/k)$-approximate matrix product property.
\end{lemma}
\begin{proof}
	
	We first show that $S$ satisfies the $O(\eps/k, 1/10, 2)$-JL moment property. Here, the $(\eps,\delta,\ell)$-JL moment property means that for any fixed $x \in \R^{n}$ with $\|x\|_2 = 1$, we have $\ex{ (\|Sx\|_2^2 - 1)^2 } \leq \eps^\ell \delta$, which will imply approximate matrix product by the results of \cite{kn14}.
	
	We prove this by induction on $q$. Let $\bar{k} = k_1$. First suppose $S = (Q \otimes T)$, where $Q \in \R^{k_1 \times n_1}$ is a count-sketch, and $T \in \R^{k' \times n'}$ is any random matrix which satisfies $\ex{\|Tx\|_2^2} = \|x\|_2^2$ ($T \in \R^{k' \times n'}$ is unbiased), and $\ex{(\|Tx\|_2 - 1)^2} \leq 1 + c/\bar{k}$ for some value $c < \bar{k}$. Note that both of these properties are satisfied with $c = 4$ if $T \in \R^{k_2 \times n_2}$ is itself a count-sketch matrix \cite{cw13}. Moreover, these are the only properties we will need about $T$, so we will. We now prove that $\ex{\|(S \otimes T)x\|_2^2} = 1$ and  $\ex{\|(S \otimes T)x\|_2^4}  \leq 1 + (c+4)/\bar{k}$ for any unit vector $x$.
	
	Fix any unit $x \in \R^{n}$ now (here $n = n_1 n')$, and let $x^j \in \R^{n'}$ be the vector obtained by restricted $x$ to the coordinates $jn_1+1$ to $(j+1)n_1$. For any $i \in [k_1], j \in [k']$, let $i_j = (i-1) k' + j$. Let $h_Q(i) \in [k_1]$ denote the row where the non-zero entry in the $i$-th column is placed in $Q$. Let $\sigma_Q(i) \in \{1,-1\}$ denote the sign of the entry $Q_{h_Q(i),i}$. Let $\delta_Q(i,j)$ indicate the event that $h_Q(i) = j$. First note that 
	\begin{align*}
	\bex{ \sum_{i,j}((Q \otimes T) x)_{i_j}^2} 
	& = \bex{\sum_{i=1}^{k_1} \sum_{j=1}^{k'} \left( \sum_{\tau =1 }^{n_1} \delta_Q(\tau,i) \sigma_Q(\tau) (T x^\tau)_j \right)^2 }\\
	& = \bex{\sum_{i=1}^{k_1} \sum_{j=1}^{k'}\sum_{\tau =1 }^{n_1} \delta_Q(\tau,i) (T x^\tau)_j^2 } \\
	& = \bex{\sum_{\tau =1 }^{n_1}  \sum_{i=1}^{k_1} \sum_{j=1}^{k'}\delta_Q(\tau,i) (T x^\tau)_j^2 } \\
	& = \bex{\sum_{\tau =1 }^{n_1} \|T x^\tau\|_2^2 } \\ 
	& = \|x\|_2^2
	\end{align*}
	Where the last equality follows because count-sketch $T$ is unbiased for the base case, namely that $\ex{\|Tx\|_2^2}= \|x\|_2^2$ for any $x$ \cite{w14}, or by induction. We now compute the second moment, 
	\begin{align*}
	\bex{ \left(\sum_{i,j}((Q \otimes T)x)_{i_j}^2 \right)^2 } &= \bex{ \left(\sum_{i,j} \left( \sum_{\tau =1 }^{n_1} \delta_Q(\tau,i) \sigma_Q(\tau) (T x^\tau)_j \right)^2 \right)^2 }  \\ 
	&= \bex{ \left(\sum_{i,j}\sum_{\tau_1,\tau_2 }^{n_1} \delta_Q(\tau_1,i) \sigma_Q(\tau_1) (T x^{\tau_1})_j  \delta_Q(\tau_2,i) \sigma_Q(\tau_2) (T x^{\tau_2})_j\right)^2 }  \\ 
	&=\sum_{\tau_1,\tau_2,\tau_3,\tau_4 }^{n_1}  \mathbb{E}\left[\left(\sum_{i,j} \delta_Q(\tau_1,i) \sigma_Q(\tau_1) (T x^{\tau_1})_j  \delta_Q(\tau_2,i) \sigma_Q(\tau_2) (T x^{\tau_2})_j \right) \right.  \\
	& \left. \cdot \left(\sum_{i,j} \delta_Q(\tau_3,i) \sigma_Q(\tau_3) (T x^{\tau_3})_j  \delta_Q(\tau_4,i) \sigma_Q(\tau_4) (T x^{\tau_4})_j \right)\right]  .
	\end{align*} 
	We now analyze the above expectation. There are several cases for the expectation of each term. First, we bound the sum of the expectations when $t_1 = t_2 = t_3 = t_4$ by 
	\begin{align*}
	&\sum_{\tau =1}^{n_1}  \mathbb{E}\left[\left(\sum_{i,j} \delta_Q(\tau,i) \sigma_Q(\tau) (T x^{\tau})_j  \delta_Q(\tau,i) \sigma_Q(\tau) (T x^{\tau})_j \right) \right.  \\
	& \left. \cdot \left(\sum_{i,j} \delta_Q(\tau,i) \sigma_Q(\tau) (T x^{\tau})_j  \delta_Q(\tau,i) \sigma_Q(\tau) (T x^{\tau})_j \right)\right]  \\
	\leq & \sum_{\tau =1}^{n_1}  \mathbb{E}\left[\|T x^{\tau}\|_2^4\right]  = 1 + c/\bar{k} 
	\end{align*} 
	Where the last equation follows from the variance of count-sketch \cite{cw13} for the base case, or by induction for $q \geq 3$. 
	We now bound the sum of the expectations when $t_1 = t_2 \neq t_3 = t_4$ by 
	\begin{align*}
	&\sum_{\tau_1 \neq \tau_2}  \mathbb{E}\left[\left(\sum_{i,j} \delta_Q(\tau_1,i) \sigma_Q(\tau_1) (T x^{\tau_1})_j  \delta_Q(\tau_1,i) \sigma_Q(\tau_1) (T x^{\tau_1})_j \right) \right.  \\
	& \left. \cdot \left(\sum_{i,j} \delta_Q(\tau_2,i) \sigma_Q(\tau_2) (T x^{\tau_2})_j  \delta_Q(\tau_2,i) \sigma_Q(\tau_2) (T x^{\tau_2})_j \right)\right]  \\ 
	&\leq \sum_{\tau_1 \neq \tau_2} \ex{\|Tx^{\tau_1}\|_2^2   \|Tx^{\tau_2}\|_2^2/k_1} \\
	& \leq \ex{\|Tx\|_2^4 /k_1} \leq (1 + c/\bar{k})/k_1.
	\end{align*} 
	We can similarly bound the sum of the terms with $t_1 = t_3 \neq t_2 = t_4$ and $t_1 = t_4 \neq t_3 = t_2$ by $ (1 + c/\bar{k})/k_1$, giving a total bound on the second moment of $$ \ex{\|(Q \otimes T)x\|_2^4  }\leq 1 +c/\bar{k} +  3(1 + c/\bar{k})/k_1) \leq 1 + (4+ c)/\bar{k}$$ since any term with a $t_i \notin \{t_1 , t_2, t_3 ,t_4\} \setminus \{t_i\}$ immediately has expectation $0$. By induction, it follows that $\ex{(\otimes_{i=1}^q S_i) x\|_2^2} = 1$ for any unit $x$, and $\ex{(\otimes_{i=1}^q S_i) x\|_2^4} \leq 1 + (4q + c)/\bar{k}$, where $c$ is the constant from the original variance of count-sketch. Setting $\bar{k} = k_1 = \dots = k_q = \Theta(q k^2 / \eps^2)$ with a large enough constant, this completes the proof that $S = (\otimes_{i=1}^q S_i)$ has the  $O(\eps/k, 1/10, 2)$-JL moment property. Then by Theorem 21 of \cite{kn14}, we obtain the approximate matrix product property:
	$$ \pr{\| A^\top  S^\top S B - A^\top  B\|_F \leq O(\eps /k) \|A\|_F \|B\|_F} \geq 9/10$$ 
	for any two matrices $A,B$. Letting $A = B^\top = U$ where $U \in \R^{n \times k}$ is a orthogonal basis for any $k$-dimensional subspace $\mathcal{V} \subset \R^n$, it follows that
	\begin{align*}
	\| U^\top  S^\top S U - I_k \|_F \leq O(\eps/k) \|U\|_F^2 \leq O(\eps),
	\end{align*} 
	where the last step follows because $U$ is orthonormal, so $\|U\|_F^2 = k$. Since the Frobenius norm upper bounds the spectral norm $\|\cdot\|_2$, we have $\| U^\top S^\top S U - I_k \|_2 \leq O(\eps)$, from which it follows that 
	all the eigenvalues of  $U^\top  S^\top  S U$ are in $(1-O(\eps),1+O(\eps))$, which implies $\|S Ux\|_2 = (1 \pm O(\eps))\|x\|_2$ for all $x \in \R^n$, so for any $y \in \mathcal{V}$, let $x_y$ be such that $y = Ux_y$, and then 
	\begin{align*}
	\|Sy\|_2 = \|SUx_y\|_2=(1 \pm O(\eps))\|x_y\|_2 =(1 \pm O(\eps))\|Ux_y\|_2 =(1 \pm O(\eps))\|y\|_2,
	\end{align*}
	which proves that $S$ is a subspace embedding for $\mathcal{V}$ (not the second to last inequality holds because $U$ is orthonormal).
\end{proof}
Finally, we are ready to prove Lemma \ref{lem:pcp}.

\begin{lemma}\label{lem:pcp}
	Let $S =(\otimes_{i=1}^q S_i)$ with $S_i \in \R^{k_i \times n_i}$, $k_1 = k_2 = \dots = k_q =  \Theta(qk^2/\eps^2)$. Then with probability $9/10$ $SA$ is a \textit{Projection Cost Preserving Sketch} (\textsc{PCPSketch}) for $A$, namely for all rank $k$ orthogonal projection matrix $P \in \R^{d \times d }$, 
	\[(1-\eps)\|A - AP\|_F^2 \leq \|SA - SA P\|_F^2 + c \leq (1+\eps)\|A - AP\|_F^2	\]
	where $c \geq 0$ is some fixed constant independent of $P$ (but may depend on $A$ and $SA$).
\end{lemma}
\begin{proof}
	To demonstrate that $SA$ is a \textsc{PCPSketch}, we show that the conditions of Lemma 10 of \cite{cemmp15} hold, which imply this result. Our result follows directly from Theorem 12 of \cite{cemmp15}. Note that all that is needed (as discussed below the theorem) for the proof is that $S$ is an $\eps$-subspace embedding for a fixed $k$-dimensional subspaces, and that $S$ satisfies the $(\eps/\sqrt{k})$ approximate matrix product property. By Lemma \ref{lem:subspace_amp}, we have both $\eps$-subspace embedding for $S$ as well as a stronger $(\eps/k)$ approximate matrix product property. Thus Theorem 12 holds for the random matrix $S$ when $k_1 = k_2 = \dots = k_q = \Theta(qk^2/\eps^2)$, which completes the proof.

\end{proof}

\section{Numerical Simulations}


In our numerical simulations, we compare our algorithms to two baselines: (1) brute force, i.e., directly solving regression without sketching, and (2) the methods based sketching developed in \cite{dssw18}. All methods were implemented in Matlab on a Linux machine. We remark that in our implementation, we simplified some of the steps of our theoretical algorithm, such as the residual sampling algorithm (Alg. \ref{alg:residualsample}). We found that in practice, even with these simplifications, our algorithms already demonstrated substantial improvements over prior work.
	

Following the experimental setup in \cite{dssw18}, we generate matrices $A_1 \in \mathbb{R}^{300\times 15}$, $A_2 \in \mathbb{R}^{300\times 15}$, and $b\in\mathbb{R}^{300^2}$, such that all entries of $A_1,A_2,b$ are sampled i.i.d. from a normal distribution. Note that $A_1\otimes A_2 \in \mathbb{R}^{90000\times 225}$.  We define $T_{\mathrm{bf}}$  to be the time of the brute force algorithm, $T_{\mathrm{old}}$ to be the time of the algorithms from \cite{dssw18}, and $T_{\mathrm{ours}}$ to be the time of our algorithms.  We are interested in the time ratio with respect to the brute force algorithm and the algorithms from \cite{dssw18}, defined as, $r_t = T_{\mathrm{ours}}/T_{\mathrm{bf}}$, and $r'_t = T_{\mathrm{ours}}/T_{\mathrm{old}}$. The goal is to show that our methods are significantly faster than both baselines, i.e., both $r_t$ and $r'_t$ are significantly less than $1$.  

We are also interested in the quality of the solutions computed from our algorithms, compared to the brute force method and the method from \cite{dssw18}. Denote the solution from our method as $x_{\mathrm{our}}$, the solution from the brute force method as $x_{\mathrm{bf}}$, and the solution from the method in \cite{dssw18} as $x_{\mathrm{old}}$. We define the relative residual percentage $r_e $and $r'_e$ to be:
\[r_e = 100 \frac{\lvert \|\mathcal{A} x_{\mathrm{ours}} - b\| - \|\mathcal{A}x_{\mathrm{bf}} - b \|\rvert}{\|\mathcal{A} x_{\mathrm{bf}} - b \|}, \; \; \; r'_e = 100 \frac{\lvert \|\mathcal{A} x_{\mathrm{old}} - b\| - \|\mathcal{A}x_{\mathrm{bf}} - b \|\rvert}{\|\mathcal{A} x_{\mathrm{bf}} - b \|}\]
	Where $\mathcal{A} = A_1 \otimes A_2$. The goal is to show that $r_e$ is close zero, i.e., our approximate solution is comparable to the optimal solution in terms of minimizing the error $\|\mathcal{A}x - b\|$. 

Throughout the simulations, we use a moderate input matrix size so that we can accommodate the brute force algorithm and to compare to the exact solution. We consider varying values of $m$, where $M$ denotes the size of the sketch (number of rows) used in either the algorithms of \cite{dssw18} or the algorithms in this paper. We also include a column $m/n$ in the table, which is the ratio between the size of the sketch and the original matrix $A_1 \otimes A_2$. Note in this case that $n = 90000$. 

\paragraph{Simulation Results for $\ell_2$}

{\small
\begin{table}
\caption{Results for $\ell_2$ and $\ell_1$-regression with respect to different sketch sizes $m$.}
\label{ta:1}
\begin{center}
	\begin{tabular}{|c|c|c|c|c|c|c|}
\hline
	 & $m$ & $m/n$ & $r_e$ &  $r'_e$ & $r_t$  & $r'_t$\\
	\hline 
 \multirow{3}{1em}{$\ell_2$} 
 & 8100 &.09 & 2.48\% & 1.51\%  & 0.05 & 0.22 \\
  & 12100 & .13 & 1.55\%  & 0.98\% &  0.06 & 0.24   \\
  & 16129 & .18 &  1.20\% & 0.71\% & 0.07 & 0.08   \\
	\hline
	\multirow{3}{1em}{$\ell_1$} & 2000 & .02 &7.72\% & 9.10\% & 0.02 & 0.59 \\
	& 4000 & .04 & 4.26\% & 4.00\% & 0.03 & 0.75 \\
	& 8000 & .09 & 1.85\% & 1.6\% & 0.07 & 0.83 \\
	& 12000 &  .13 & 1.29\%  & 0.99\% &  0.09 & 0.79  \\
	& 16000 & .18 & 1.01\% & 0.70\% & 0.14 &  0.90   \\
		\hline
\end{tabular}
\end{center}
\label{tab:l2}
\end{table}
}

We first compare our algorithm, Alg.~\ref{alg:l2reg}, to baselines under the $\ell_2$ norm.  In our implementation, $\min_{x}\|A x - b\|_2$ is solved by Matlab backslash $A\textbackslash b$.
 Table~\ref{tab:l2} summarizes the comparison between our approach and the two baselines. The numbers are averaged over $5$ random trials. First of all, we notice that our method in general provides slightly less accurate solutions than the method in \cite{dssw18}, i.e., $r_e > r'_e$ in this case. However, comparing to the brute force algorithm, our method still generates relatively accurate solutions, especially when $m$ is large, e.g., the relative residual percentage w.r.t. the optimal solution is around $1\%$ when $m \approx 16000$. On the other hand, as suggested by our theoretical improvements for $\ell_2$, our method is significantly faster than the method from \cite{dssw18}, consistently across all sketch sizes $m$. Note that when $m \approx 16000$, our method is around $10$ times faster than the method in \cite{dssw18}. For small $m$, our approach is around $5$ times faster than the method in \cite{dssw18}.

\paragraph{Simulation Results for $\ell_1$}


We compare our algorithm, Alg.~\ref{alg:l1}, to two baselines under the $\ell_1$-norm. The first is a brute-force solution, and the second is the algorithm for \cite{dssw18}. For $\min_x \|Ax-b\|_1$, the brute for solution is obtained via a Linear Programming solver in Gurobi \cite{gurobi}.   
Table~\ref{tab:l2} summarizes the comparison of our approach to the two baselines under the $\ell_1$-norm. The statistics are averaged over $5$ random trials.  Compared to the Brute Force algorithm, our method is consistently around $10$ times faster, while in general we have relative residual percentage around 1\%.  Compared to the method from \cite{dssw18}, our approach is consistently faster (around $1.3$ times faster). Note our method has slightly higher accuracy than the one from \cite{dssw18} when the sketch size is small, but slightly worse accuracy when the sketch size increases.



\section{Entry-wise Norm Low trank Approximation}\label{app:lowTrank}


We now demonstrate our results for low $\trank$ approximation of arbitrary input matrices. Specifically, 
 we study the following problem, defined in \cite{l00}: given $A \in \R^{n^2 \times n^2}$, the goal is to output a $\trank$-$k$ matrix $B \in \R^{n^2 \times n^2}$ such that
\begin{align}
\| B - A \|_{\xi} \leq \alpha \cdot \OPT .
\label{eq:trank_k}
\end{align}
for some $\alpha \geq 1$, where 
$
\OPT = \min_{\trank-k~A'} \| A' - A \|_{\xi},
$, 
where the $\trank$ of a matrix $B$ is defined as the smallest integer $k$ such that $B$ can be written as a summation of $k$ matrices, where each matrix is the Kronecker product of $q$ matrices with dimensions $n \times n$:  $B = \sum_{i=1}^k U_i \otimes V_i$,
where $U_i, V_i \in \R^{n \times n}$.  

Using Lemma~\ref{lem:reshaping}, we can rearrange the entries in $A \in \R^{n^2 \times n^2}$ to obtain $\ov{A} \in \R^{n^2 \times n^2}$, where the $(i+n(j-1))$'th row of $\bar{A}$ is equal to $vec((A_1)_{i,j}A_2)$, and also vectorize the matrix $U_i \in \R^{n \times n}$ and $V_i \in \R^{n \times n}$ to obtain vectors $u_i \in \R^{n^2}, v_i \in \R^{n^2}$. Therefore, for any entry-wise norm $\xi$ we have
\begin{align*}
\left\| \sum_{i=1}^k U_i \otimes V_i - A \right\|_{\xi} = \left\| \sum_{i=1}^k u_i v_i^\top - \ov{A} \right\|_{\xi}
\end{align*}

\begin{lemma}[Reshaping for Low Rank Approximation]\label{lem:reshaping_low_rank}
	There is a one-to-one mapping $\pi : [n] \times [n] \times [n] \times [n] \rightarrow [n^2] \times [n^2]$ such that for any pairs $(U,u) \in \R^{n \times n} \times \R^{n^2}$ and $(V,v) \in \R^{n \times n} \times \R^{n^2}$, if $U_{i_1,j_1} = u_{i_1 + n ( j_1-1)}$ and $V_{i_1,j_1} = v_{i_1 + n ( j_1-1)}$, then we have for $i_1,i_2,j_1,j_2$
	\begin{align*}
	(U \otimes V)_{ i_1 + n (i_2-1) , j_1 + n(j_2 -1) } = ( u \cdot v^\top )_{ \pi(i_1,i_2,j_1,j_2) }
	\end{align*}
	where $U \otimes V \in \R^{n^2 \times n^2}$ and $u v^\top \in \R^{n^2 \times n^2}$.
\end{lemma}
\begin{proof}
	
	We have{
		\begin{align*}
		(U \otimes V)_{ i_1 + n (i_2-1) , j_1 + n( j_2 -1 ) } 
		= & ~ U_{i_1,j_1} V_{i_2,j_2} \\
		= & ~ u_{i_1 + n ( j_1-1)} \cdot v_{i_2 +  n( j_2 -1 ) }\\
		= & ~ (u v^\top)_{ i_1 +  n ( j_1 - 1 ) , i_2 + n ( j_2 - 1 ) }
		\end{align*}}
	where the first step follows from the definition of $\otimes$ product, the second step follows from the connection between $U,V$ and $u,v$, the last step follows from the outer product.
\end{proof}

Therefore, instead of using $\trank$ to define low-rank approximation of the $\otimes$ product of two matrices, we can just use the standard notion of rank to define it since both $B$ and $A'$ can be rearranged to have rank $k$. 

\begin{definition}[Based on Standard Notion of Rank]
	Given two matrices $A_1, A_2 \times \R^{n \times n}$, let $\ov{A} \in \R^{n^2 \times n^2}$ denote the re-shaping of $A_1 \otimes A_2$. The goal is to output a rank-$k$ matrix $\ov{B}$ such that
	\begin{align*}
	\| \ov{B} - \ov{A} \|_{\xi} \leq \alpha \OPT_{\xi,k}
	\end{align*}
	where $\OPT_{\xi,k} = \min_{\rank-k~\ov{A}'} \| \ov{A}' - \ov{A} \|_{\xi}$.
\end{definition}
In other words, $\ov{B}$ can be written as $\ov{B} = \sum_{i=1}^k u_i v_i^\top$ where $u_i , v_i$ are length $n^2$ vectors.

\subsection{Low-rank Approximation Results}

Combining the low-rank reshaping Lemma~\ref{lem:reshaping_low_rank} with the main input-sparsity low-rank approximation of \cite{cw13}, we obtain our Frobenius norm low rank approximation result.
\begin{theorem}[Frobenius norm low rank approximation, $p=2$]
	For any $\epsilon \in (0,1/2)$, there is an algorithm that runs in $n^2 \poly(k/\epsilon)$ and outputs a rank-$k$ matrix $\ov{B}$ such that
	\begin{align*}
	\| \ov{B} - \ov{A} \|_F \leq (1+\epsilon) \OPT_{F,k}
	\end{align*}
	holds with probability at least $9/10$, where $\OPT_p$ is cost achieved by best rank-$k$ solution under the $\ell_p$-norm.
	\label{thm:forbenius_norm}
\end{theorem}

Similarly, using the main $\ell_p$ low rank approximation algorithm of \cite{swz17}, we have
\begin{theorem}[Entry-wise $\ell_p$-norm low rank approximation, $1 \leq p \leq 2$]
	There is an algorithm that runs in $n^2 \poly(k)$ and outputs a rank-$k$ matrix $\ov{B}$ such that
	\begin{align*}
	\| \ov{B} - \ov{A} \|_p \leq \poly(k\log n) \OPT_{p,k}
	\end{align*}
	holds with probability at least $9/10$, where $\OPT_p$ is cost achieved by best rank-$k$ solution under the $\ell_p$-norm.
	\label{thm:lp}
\end{theorem}

Applying the bi-criteria algorithm of \cite{cgklpw17} gives us:
\begin{theorem}[General $p > 1$, bicriteria algorithm]
	There is an algorithm that runs in $\poly(n,k)$ and outputs a rank-$\poly(k\log n)$ matrix $\ov{B}$ such that
	\begin{align*}
	\| \ov{B} - \ov{A} \|_p \leq \poly(k\log n) \OPT_p
	\end{align*}
	holds with probability at least $9/10$, where $\OPT_{p,k}$ is cost achieved by best rank-$k$ solution under the $\ell_p$-norm.
\end{theorem}

Finally using the low-rank approximation algorithm for general loss functions given in \cite{swz18}, we obtain a very general result. The parameters for the loss function described in the following theorem are discussed in Section \ref{app:proplowTrank}.
\begin{theorem}[General loss function $g$]
	For any function $g$ that satisfies Definition~\ref{def:intro_approximate_triangle_inequality}, \ref{def:intro_monotone_property}, \ref{def:intro_regression_property}, there is an algorithm that runs in $O(n^2 \cdot T_{\reg,g,n^2,k,n^2})$ time and outputs a $\rank$-$O(k \log n)$ matrix $\ov{B} \in \R^{n^2 \times n^2}$ such that
	\begin{align*}
	\| \ov{B} - \ov{A} \|_g \leq \ati_{g,k} \cdot \sym_{g} \cdot \reg_{g,k} \cdot O(k\log k) \cdot \OPT_{g,k},
	\end{align*}
	holds with probability $1-1/\poly(n)$.
	\label{thm:general_p}
\end{theorem}

Hence, overall, the strategy is to first reshape $A = A_1\otimes A_2$ into $\bar{A}$, then compute $\bar{B} = \sum_{i=1}^k u_i v_i^{\top}$ using any of the above three theorems depending on the desired norm, and finally reshape $u_i$ and $v_i$ back to $U_i\in\mathbb{R}^{n\times n}$ and $V_i\in\mathbb{R}^{n\times n}$. It is easy to verify that the guarantees from Theorems~\ref{thm:general_p},~\ref{thm:lp},~\ref{thm:forbenius_norm} are directly transferable to the guarantee of the $\trank-k$ approximation shown in Eq~\ref{eq:trank_k}.

\subsection{Properties for General Loss Functions}\label{app:proplowTrank} 

We restate three general properties (defined in \cite{swz18}), the first two of which are structural properties and are necessary and sufficient for obtaining a good approximation from a small subset of columns. The third property is needed for efficient running time. 

\begin{definition}[Approximate triangle inequality]\label{def:intro_approximate_triangle_inequality}
For any positive integer $n$, we say a function $g(x) : \R \rightarrow \R_{\geq 0}$ satisfies the ${\ati}_{g,n}$-approximate triangle inequality if for any $x_1, x_2, \cdots, x_n \in \R$ we have
\begin{align*}
g \left( \sum_{i=1}^n x_i \right) \leq {\ati}_{g,n} \cdot \sum_{i=1}^n g( x_i ).
\end{align*}
\end{definition}

\begin{definition}[Monotone property]\label{def:intro_monotone_property}
For any parameter $\sym_g \geq 1$, we say function $g(x) : \R \rightarrow \R_{\geq 0}$ is $\sym_g$-monotone if
for any $x,y\in\mathbb{R}$ with $0\leq |x|\leq |y|,$ we have $g(x)\leq \sym_{g}\cdot g(y).$
\end{definition}

\begin{definition}[Regression property]\label{def:intro_regression_property}
We say function $g(x) : \R \rightarrow \R_{\geq 0}$ has the $(\reg_{g,d},T_{\reg,g,n,d,m})$-regression property if the following holds: given two matrices $A \in \R^{n \times d}$ and $B \in \R^{n \times m}$, for each $i\in [m]$, let $\OPT_i$ denote $\min_{x \in \R^{ d } } \| A x - B_i \|_g $. There is an algorithm that runs in $T_{\reg,g,n,d,m}$ time and outputs a matrix $X'\in \R^{d \times m}$ such that
\begin{align*}
 \| A X'_i - B \|_g \leq \reg_{g,d} \cdot \OPT_i, \forall i \in [m] 
\end{align*}
and outputs a vector $v\in \R^d$ such that
\begin{align*}
\OPT_i \leq v_i \leq \reg_{g,d} \cdot \OPT_i, \forall i \in [m].
\end{align*}
The success probability is at least $1-1/\poly(nm)$.
\end{definition}


\section*{Acknowledgments}
The authors would like to thank Lan Wang and Ruosong Wang for a helpful discussion. The authors would like to thank Lan Wang for introducing All-Pairs Regression problem to us.


\newpage
\bibliographystyle{alpha}
\bibliography{ref}




\end{document}